\documentclass[12pt,centertags,reqno]{amsart}

\usepackage[foot]{amsaddr}
\usepackage{latexsym}
\usepackage[english]{babel}
\usepackage[T1]{fontenc}
\usepackage[numbers]{natbib}
\usepackage{amssymb}
\usepackage{fancyhdr}
\usepackage{url}
\usepackage{hyperref}
\usepackage{verbatim}
\usepackage{leftidx}
\usepackage{color,graphicx}

\usepackage{mathrsfs, mathtools}
\usepackage{stmaryrd}

\usepackage{marvosym}
\mathtoolsset{showonlyrefs}

\numberwithin{equation}{section} \makeatletter
\renewcommand{\subsection}{\@startsection
{subsection}{2}{0mm}{\baselineskip}{-0.25cm}
{\normalfont\normalsize\bf}} \makeatother

\newtheorem{theorem}{Theorem}[section]
\newtheorem{lemma}[theorem]{Lemma}
\newtheorem{corollary}[theorem]{Corollary}
\newtheorem{definition}[theorem]{Definition}
\newtheorem{remark}[theorem]{Remark}
\newtheorem{proposition}[theorem]{Proposition}

\newtheorem{ass}[theorem]{Assumption}

\def \C {\mathcal C}
\def \E {\mathcal E}
\def \F {\mathcal F}
\def \G {\mathcal G}
\def \H {\mathcal H}
\def \L {\mathcal L}

\def \P {\mathbf P}
\def \Q {\mathbf Q}
\def \I {{\mathbf 1}}

\def \R {\mathbb R}
\def \bF {\mathbb F}
\def \bG {\mathbb G}
\def \bH {\mathbb H}
\def \bE {\mathbb E}

\newcommand{\ud}{\mathrm d}
\newcommand{\ds}{\displaystyle}
\newcommand{\esp}[2][\mathbb E] {#1\left[#2\right]}
\newcommand{\widehatesp}[2][\widehat{\mathbb E}] {#1\left[#2\right]}
\newcommand{\espp}[2][\mathbb {\widehat E}] {#1\left[#2\right]}

\newcommand{\doleans}[1] {\mathcal E\left(#1\right)}

\begin{document}
\author[C.~Ceci]{Claudia  Ceci}
\address{Claudia  Ceci ({\large \Letter}), Department of Economics,
University ``G. D'Annunzio'' of Chieti-Pescara, Viale Pindaro, 42,
I-65127 Pescara, Italy.}\email{c.ceci@unich.it}

\author[K.~Colaneri]{Katia Colaneri}
\address{Katia Colaneri, Department of Economics,
 University of Perugia, Via Alessandro Pascoli, 20, I-06123 Perugia, Italy.}\email{katia.colaneri@unipg.it}

\author[A.~Cretarola]{Alessandra Cretarola}
\address{Alessandra Cretarola, Department of Mathematics and Computer Science,
 University of Perugia, Via Luigi Vanvitelli, 1, I-06123 Perugia, Italy.}\email{alessandra.cretarola@unipg.it}

\title[Unit-linked life insurance policies]{Unit-linked life insurance policies: optimal hedging in partially observable market models}

\date{}

\begin{abstract}
\begin{center}
In this paper we investigate the hedging problem of a unit-linked life insurance contract via the local risk-minimization approach, when the insurer has a restricted information on the market. In particular, we consider an endowment insurance contract, that is a combination of a term insurance policy and a pure endowment,
whose final value depends on the trend of a stock market where the premia the policyholder pays are invested.
We assume that the stock price process dynamics depends on an exogenous unobservable stochastic factor that also influences the mortality rate of the policyholder.
To allow for mutual dependence between the financial and the insurance markets, we use the progressive enlargement of filtration approach. We characterize the optimal hedging strategy in terms of the integrand in the Galtchouk-Kunita-Watanabe decomposition of the insurance claim with respect to the minimal martingale measure and the available information flow. We provide an explicit formula by means of predictable projection of the corresponding hedging strategy under full information with respect to the natural filtration of the risky asset price and the minimal martingale measure.
Finally, we discuss applications in a Markovian setting via filtering.
\end{center}
\end{abstract}

\maketitle




\section{Introduction}
For the last years unit-linked life insurance contracts have experienced a clamorous success, driven by low interest rates, which have considerably reduced the returns of the classic management, and the new Solvency II rules on the insurance regulatory capital, which made the unit-linked much more affordable for the companies, in terms of lower absorption of capital.
In these insurance products premia are invested by the insurance company in the financial market on behalf on the policyholder. Therefore, benefits depend on the performance of a stock or a portfolio traded in the financial market. Among these contracts, we may distinguish at least three different kinds of policies based on the payoff structure:
\begin{itemize}
\item {\em pure endowment} contract that promises to pay an agreed amount if the policyholder is still alive on a specified future date;
\item {\em term insurance} contract that pays the benefit if the policyholder dies before the policy term;
\item {\em endowment insurance} contract which is a combination of the above contracts and guarantees that benefits will be paid by the insurance company, either at the policy term or after the insured death.
\end{itemize}
Modeling the time of death is a fundamental issue to be addressed in this setting. Here, we propose a modeling framework for life insurance liabilities that is also well suited to describe defaultable claims, as the time of death can be handled in a similar manner to the default time of a firm. Then, we take the analogies between mortality and credit risk into account and follow the intensity-based approach of reduced-form methodology, see e.g. \citet{bielecki2002} and references therein.
The goal of this paper is to study the hedging problem of an endowment insurance contract in a general intensity-based model where the mortality intensity, as well as the drift in the risky asset price dynamics affecting the benefits for the policyholder, is not observable by the insurance company.
This problem requires to consider a suitable combined financial-insurance market model where we allow for mutual dependence between the stock price trend and the insurance portfolio.

Precisely, we consider a simple financial market model with a riskless asset, whose discounted price is equal to $1$, and a risky asset, with discounted price process denoted by $S$. The price process $S$  is represented by a geometric diffusion, whose drift depends on an exogenous unobservable stochastic factor $X$, correlated with $S$.
The insurance company issues an endowment insurance contract with maturity of $T$ years for an individual whose remaining lifetime is represented by a random time $\tau$.

We model the death time $\tau$ as a nonnegative random variable, which is not necessarily a stopping time with respect to the initial filtration $\bF$ generated by the underlying Brownian motions driving the dynamics of the pair $(S,X)$. We do not assume independence between the random time of death and the financial market, and characterize our setting via the progressive enlargement of filtration approach, see the seminal works by \citet{jeulin1978grossissement, jeulin1980, jeulin1985grossissements}.
This technique is widely applied to reduced-form models for credit risk, as in \citet{bielecki2004stochastic,bielecki2006completeness,bielecki2006replication, elliott2000models, kusuoka1999remark}. Moreover, applications to insurance problems can be found in \citet{biagini2015risk,barbarin2007risk,choulli2015hedging,li2011uncertain} in a complete information setting.
Here, we consider an enlargement of the filtration $\bF$ to make $\tau$ a stopping time and we denote it by $\bG$.
The available information to the insurance company is represented by a subfiltration $\widetilde \bG$ of $\bG$, which contains the natural filtration of $S$ and ensures that $\tau$ is still a stopping time.
This means that, at any time $t$, the insurer may observe the risky asset price and knows if the policyholder is still alive or not.

The endowment insurance contract can be treated as a contingent claim in the incomplete hybrid market model given by the financial securities and the insurance portfolio. Then,
we choose, among the quadratic hedging methods, the {\em local risk-minimization} approach (see e.g. \citet{schweizer2001guided} for further details). The idea of this technique is to find an optimal hedging strategy that perfectly replicates the given contingent claim with minimal cost, within a wide class of admissible strategies that in general might not necessarily be self-financing.
Locally risk-minimizing hedging strategies can be characterized via the F\"{o}llmer-Schweizer decomposition of the random variable representing the payoffs of the given contingent claim, see e.g. 
\citet{schweizer1995minimal,schweizer2001guided} for the full information case and \citet{ceci2014bsdes, ceci2015local} under incomplete information. This quadratic hedging approach has been successfully applied to the hedging problem of insurance products, see e.g. \citet{biagini2015risk,biagini2016risk,choulli2015hedging,dahl2006valuation,moller1998risk,moller2001risk,vandaele2008locally} for the complete information case and \citet{ceci2015hedging} under partial information.

To the best of our knowledge this is the first time that the hedging problem of a unit-linked life insurance policy is studied under partial information without assuming independence between the financial and the insurance markets.

Analogously to \citet{bielecki2006hedging, biagini2012local}, we assume that hedging stops after the earlier between the policyholder death $\tau$ and the maturity $T$: this allows to work with stopped price processes and guarantees that the stopped Brownian motions, that drive the financial market, are also Brownian motions with respect to the enlarged filtration. As a consequence, we do not need to assume the {\em martingale invariance property}, also known as {\em H-hypothesis}, see e.g. \citet{bielecki2002}.
Then, we introduce the (stopped) F\"{o}llmer-Schweizer decomposition under partial information and the corresponding Galtchouk-Kunita-Watanabe decomposition with respect to the minimal martingale measure and in Proposition \ref{prop:risk_min_strategy}  and  Theorem \ref{thm:equivalence}  we
characterize the optimal strategy in terms of the integrands in these decompositions. In this sense, we extend the results obtained in \citet{biagini2012local} to the partial information framework.
Moreover, Theorems  \ref{thm:strategy0}  and \ref{thm:strategy} provide the relationship between the optimal hedging strategy under partial information and that under full information via predictable projections.
In the case where the mortality intensity depends on the unobservable stochastic factor  $X$, we can compute the optimal hedging strategy in a more explicit form by means of filtering problems.

Pricing and hedging problems for contingent claims under incomplete information using filtering techniques have been studied in credit risk context, in \citet{frey2009pricing, frey2012pricing, tardelli2015partially} and in the insurance framework in \citet{ceci2015hedging} under the hypothesis of independence between the financial and the insurance markets.

The paper is organized as follows. In Section \ref{sec:market model} we introduce the combined financial-insurance market model
in a partial information scenario via progressive enlargement of filtrations. The semimartingale decompositions of the stopped risky asset price process with respect to the enlarged filtrations $\bG$ and $\widetilde \bG$ respectively, can be found in Section \ref{sec:semi_mg_decomp}. In Section \ref{sec:lrm} we provide
a closed formula for the locally risk-minimizing hedging  strategy under incomplete information for the given endowment insurance contract  by means of predictable projections. Finally, in Section \ref{sec:markov} we discuss the problem in a Markovian framework, where the mortality intensity depends on the unobservable stochastic factor and apply the filtering approach to compute the optimal hedging  strategy.
In addition, we address the issue of the hazard process and the martingale hazard process of $\tau$ under restricted information in Appendix \ref{appendix:hazard_rate}.
Some technical results on the optional and predictable projections under partial information and certain proofs %
can be found in Appendix \ref{appendix:projections}.

\section{The setting}\label{sec:market model}
We consider the problem of an insurance company that wishes to hedge a unit-linked life insurance contract. This type of contract has a relevant link with the financial market. Indeed, the value of the policy is determined by the performance of the underlying stock or portfolio. Moreover, it also depends on the remaining lifetime of the policyholder. Therefore, the most appropriate way to address the problem is to construct a combined financial-insurance market model and treat the life insurance policy as a contingent claim. We will 
define the suitable modeling framework via the progressively enlargement of filtration approach, which allows for possible dependence between the financial market and the insurance portfolio. 
First, we introduce the underlying financial market model.

\subsection{The financial market model}
Let $(\Omega,\F,\P)$ be a complete probability space endowed with a filtration $\bF=\{\F_t, \ t \in[ 0 , T]\}$, where $T$ denotes a fixed and finite time horizon, such that $\F=\F_T$ and $\F_0=\{\Omega, \emptyset\}$. On $(\Omega,\F,\P)$, we define two one-dimensional, independent $(\bF,\P)$-Brownian motions $W=\{W_t, \ t \in [0,T]\}$  and $B=\{B_t, \ t\in [0, T]\}$, with $W_0=B_0=0$.
We suppose that
\[
\bF=\bF^W\vee \bF^B,
\]
where $\bF^W$ and $\bF^B$ denote the natural filtrations of the processes $W$ and $B$, respectively.
In addition, we assume that $\bF$ satisfies the usual hypotheses of completeness and right continuity.

On the given probability space $(\Omega, \F, \P)$, we consider a simple financial market which consists of one riskless asset whose price process is assumed to be equal to $1$ at any time, and one risky asset whose (discounted) price process $S=\{S_t, \ t \in [0,T]\}$ evolves according to the following stochastic differential equation
\begin{equation}\label{eq:din_S}
\ud S_t= S_t \left(\mu(t , S_t, X_t) \ud t + \sigma(t, S_t)  \ud W_t\right),\quad S_0=s_0>0,
\end{equation}
where $X=\{X_t, \ t \in [0,T]\}$ is an unobservable exogenous stochastic factor satisfying
\begin{equation}\label{eq:din_X}
\ud X_t = b(t , X_t) \ud t  + a(t, X_t) \left[ \rho \ud W_t + \sqrt{1-\rho^2} \ud B_t\right],\quad X_0=x_0 \in \R,
\end{equation}
with $\rho \in [0,1]$. Here, the coefficients $\mu$, $b$ are $\R$-valued measurable functions and $\sigma$, $a$ are $\R^+$-valued measurable functions such that the system of equations \eqref{eq:din_S} and \eqref{eq:din_X} admits a unique strong solution, see for instance \citet[Chapter 5]{oksendal2013}.

We assume that the following conditions are in force throughout the paper:
\begin{ass}\label{ipot1}
\begin{itemize}
\item[]
\item[(i)] $\ds \esp{\int_0^{T} \left(|\mu(u, S_u, X_u)| + \sigma^2\left(u,S_u\right)\right)   \ud u} < \infty$;
\item[(ii)] $\ds \left|\frac{\mu(t, S_t, X_t)}{\sigma(t,S_t)}\right|  < c$, $\P$-a.s.,  for every $t \in [0,T]$, with $c$ being a positive constant.
\end{itemize}
\end{ass}
In particular, condition (ii) of Assumption \ref{ipot1} is required to avoid technicalities.

We observe that
$\bF^S\vee \bF^X \subseteq \bF$, where $\bF^S=\{\F_t^S, \ t \in [0,T]\}$ and $\bF^X=\{\F^X_t, \ t \in [0,T]\}$ denote the natural filtrations of the processes $S$ and $X$ respectively, and the pair $(S,X)$ is an $(\bF, \P)$-Markov process.

To exclude arbitrage opportunities, we assume that the set of all equivalent martingale measures for $S$ is non-empty and contains more than a single element, since $X$ does not represent the price of any tradeable asset, and therefore the financial market is incomplete.

Precisely, every equivalent probability measure $\Q$ has the following density $L^\Q=\{L^\Q_t, \ t \in [0,T]\}$, given by
$$
L^\Q_t:=\left.\frac{\ud \Q}{\ud \P}\right|_{\F_t} = \E \left(\int_0^. -\frac{\mu(u , S_u, X_u)}{\sigma(u, S_u) }\ud W_u + \int_0^. \psi^\Q_u \ud B_u\right)_t, \quad t \in [0,T],
$$
where $\psi^\Q=\{\psi_t^\Q,\ t \in [0,T]\}$ is an  $\bF$-predictable process such that $L^\Q$ turns out to be an
$(\bF, \P)$-martingale. Here $\E (Y)$ denotes the Dol\'{e}ans-Dade exponential of an $(\bF, \P)$-semimartingale $Y$.
The choice $\psi^\Q_t=0$, for every $t \in [0,T]$, corresponds to the so-called {\em minimal martingale measure} for $S$ (see e.g. \citet{follmer2010minimal}), denoted by $\widehat \P$, whose density process $L=\{L_t, \ t \in [0,T]\}$, is defined by
\begin{equation}\label{eq:densitaL}
L_t:=\left.\frac{\ud \widehat \P}{\ud \P}\right|_{\F_t} = \E \left(\int_0^. -\frac{\mu(u ,S_u, X_u)}{\sigma(u, S_u) }\ud W_u \right)_t, \quad t \in [0,T].
\end{equation}
Condition (ii) of Assumption \ref{ipot1} implies that $L$ is a square integrable $(\bF, \P)$-martingale. As a consequence of the Girsanov Theorem, we get that the process $\widehat W=\{\widehat W_t, \ t \in [0,T]\}$, given by
\begin{equation}\label{eq:What}
\widehat W_t:=W_t+\int_0^t\frac{\mu(u , S_u, X_u)}{\sigma(u, S_u) }\ud u, \quad t \in [0,T],
\end{equation}
is an $(\bF, \widehat \P)$-Brownian motion.

\subsection{The combined financial-insurance market model}

Let $\tau$ be the remaining lifetime of an individual with age $a$. Here $\tau$ is a nonnegative random variable $\tau:\Omega\to[0, T]\cup \{+\infty\}$ satisfying $\P(\tau=0)=0$ and $\P(\tau>t)>0$, for every $t \in [0,T]$. Since, we only consider a single policyholder we omit the dependence on the age.

Then, we define the associated death indicator process as $H=\{H_t, \ t \in [0, T]\}$, where
\begin{equation}\label{eq:death_process}
H_t=\I_{\{\tau\le t\}}, \quad t \in [ 0, T],
\end{equation}
and $\bF^H=\{\F^H_t,\ t\in [0, T]\}$ denotes the natural filtration of $H$. 
Notice that $\tau$ is a stopping time with respect to the filtration $\bF^H$, but it is not necessarily a stopping time with respect to the filtration $\bF$.

Let $\bG=\{\G_t,\ t\in [0, T]\}$ be the enlarged filtration given by
\[
\G_t:=\F_t\vee\F_t^H, \quad t \in [0,T].
\]
This is the smallest filtration which contains $\bF$, such that $\tau$ is a $\bG$-stopping time. In this framework the initial market might be correlated with the time of death $\tau$. The connection between the financial market and $\tau$ is expressed in terms of the conditional distribution of $\tau$ given $\F_t$, for every $t \in [0,T]$,  defined as the process $F=\{F_t, \ t \in [0,T]\}$  given by
\begin{equation}\label{eq:cond_distribution_F}
  F_t=\P(\tau\leq t|\F_t)=\esp{H_t|\F_t}, \quad t\in[ 0, T].
\end{equation}
Notice that, $0\leq F_t\leq 1$ for every $t \in [0,T]$. In the sequel, we will assume that $F_t<1$ for every $t \in [0,T]$; this excludes the case where $\tau$ is an $\bF$-stopping time, see e.g. \citet{bielecki2002} for further details.

In the sequel we define the so-called hazard process of the random time $\tau$.
\begin{definition}
The $\bF$-{\em  hazard process} of $\tau$ under $\P$ is the nonnegative process $\Gamma=\{\Gamma_t, \ t \in [0,T]\}$ defined by
\begin{equation}\label{eq:hazard process}
  \Gamma_t=-\ln(1-F_t),\quad t \in [0,T].
\end{equation}
\end{definition}
In this paper we assume that $\Gamma$ has a density, i.e. $\Gamma_t=\int_0^t \gamma_u \ud u$, for every $t \in [0,T]$, for some nonnegative $\bF$- predictable 
 process  $\gamma=\{\gamma_t, \ t \in [0,T]\}$  such that $\esp{\int_0^T \gamma_u \ud u} < \infty$. The process $\gamma$ is known as the $\bF$-{\em mortality intensity} or the $\bF$-{\em mortality rate} and the $\bF$-{\em survival process} is given by $\P( \tau > t |\F_t) = e ^{- \int_0^t \gamma_u \ud u},$ $t \in [0,T]$.

Now, we  introduce the concept of the $(\bF,\bG)$-martingale hazard process associated to $\tau$.
\begin{definition}
An $\bF$-predictable, right-continuous, increasing process $\Lambda=\{\Lambda_t, \ t \in [0,T]\}$ is called an $(\bF,\bG)$-{\em martingale hazard process} of the random time $\tau$ if and only if
the process
\[
M_t=H_t-\Lambda_{t\wedge\tau}, \quad t \in [0,T],
\]
is a $(\bG, \P)$-martingale.
\end{definition}

\begin{remark}\label{MG}
It is well known that in general, the $\bF$-hazard process and the $(\bF,\bG)$-martingale hazard process do not coincide. Nevertheless, the existence of the $\bF$-mortality intensity ensures that the process $F$ is continuous and increasing. Then, by \citet[Proposition 6.2.1]{bielecki2002} we get that $\Gamma$ is also an $(\bF, \bG)$-martingale hazard process, and consequently,
the process $M=\{M_t, \ t \in [0,T]\}$ defined by
\begin{equation}\label{eq:martingala_salto}
M_t=H_t-\Gamma_{t\wedge\tau}=H_t-\int_0^{t\wedge \tau} \gamma_u \ud u=H_t-\int_0^{t} \lambda_u \ud u, \quad t \in [0,T],
\end{equation}
where $\lambda_t=\gamma_t\I_{\{\tau \geq t\}} = \gamma_t (1 - H_{t^-})$, is a $(\bG, \P)$-martingale.
Furthermore, by \citet[Chapter 6.78]{dellacherie_meyer1982}, $\tau$ is a totally inaccessible $\bG$-stopping time.
\end{remark}

We assume that the insurance company issues a unit-linked life insurance policy. In these contracts the insurance benefits depend on the price of some specific traded stock on the financial market, as well as the remaining lifetime of the policyholder. Therefore, the insurer is exposed to both financial and mortality risks. Precisely, we consider an endowment insurance contract with maturity of $T$ years which can be defined as follows.

\begin{definition}\label{def:dafaultable_claim}
An endowment insurance contract is characterized by a a triplet $(\xi, Z, \tau)$, where
\begin{itemize}
\item the random variable $\xi\in L^2(\F^S_T,\P)$ is the amount paid at maturity $T$, if the policyholder is still alive at time $T$;
\item the process $Z=\{Z_t,\ t \in [0,T]\}$ represents the amount which is immediately paid at death-time $\tau$; here, $Z$ is
assumed to be square integrable and $\bF^S$-predictable;
\item $\tau$ is time of death.
\end{itemize}
\end{definition}

\begin{remark}
If $Z=0$ the endowment insurance contract reduces to the so-called {\em term insurance contract}, which pays out the amount $\xi$ in case of survival until $T$, whereas, if $\xi=0$ we obtain the payoff of a {\em pure endowment contract}, that provides the amount $Z_\tau$ at the random time $\tau$ in case of death before time $T$.
\end{remark}

We denote by $N = \{N_t,\ t \in [0,T]\}$ the process that models the payment stream arising from the endowment insurance contract, i.e.
\begin{equation}\label{N}
N_t = Z_\tau \I_{\{ \tau \leq t\}}=\int_0^t Z_s \ud H_s,  \quad 0\leq t < T, \ \mbox{ and } \ N_T = \xi  \I_{\{ \tau > T \}}, \quad t =T.
\end{equation}

\subsection{The information levels}

We consider a scenario where the insurance company does not have a complete information on the market. Precisely, we assume that it  cannot observe neither the stochastic factor $X$ affecting the behavior of the risky asset price process $S$ nor the Brownian motions $W$ and $B$ which drive the dynamics of the pair $(S,X)$. In particular, this implies that the insurer does not know completely the $\bF$-mortality rate $\gamma$ of $\tau$.  For instance, $\gamma$ may be dependent on the unobservable stochastic factor $X$, that is $\gamma_t= \gamma(t, X_t)$, for each $t \in [0,T]$,
with $\gamma$ being a nonnegative measurable function. This special case will be discussed in Section \ref{sec:markov}.
At any time $t$, the insurer may observe the risky asset price and knows if the policyholder died or not. Hence, the available information is described by the filtration
$\widetilde \bG=\{\widetilde\G_t, \ t \in [0,T]\}$, given by
\[\widetilde \G_t:=\F_t^S\vee\F^H_t,\quad t \in [0,T].\]

Since $\bF^S\subseteq \bF$, we have $$\widetilde \bG \subseteq  \bG.$$

We assume throughout the paper that all filtrations satisfy the usual hypotheses of completeness and right-continuity.
Some results about the hazard process and the martingale hazard process of $\tau$ under partial information can be found in Appendix \ref{appendix:hazard_rate}.

In the sequel we will address the hedging problem of the endowment insurance contract $(\xi, Z, \tau)$ in a partial information setting characterized by the information flow  $\widetilde \bG$. Since hedging stops either at time $T$ or $\tau$, whichever comes first, it makes sense to consider the stopped discounted price process. This also implies that we can work without assuming the so-called martingale invariance property between filtrations $\bF$ and $\bG$, which establishes that  every $\bF$-martingale is also a $\bG$-martingale. The martingale invariance property is frequently assumed when considering enlargement of filtrations. To the best of our knowledge there are only a few papers in the literature where this hypothesis is not imposed, see for instance \citet{barbarin2009valuation, choulli2015hedging} in the insurance framework and \citet{biagini2012local} in the credit risk setting.

\section{The semimartingale decompositions of the stopped risky asset price process} \label{sec:semi_mg_decomp}

In this section we provide the semimartingale decomposition of the stopped process price process $S^\tau = \{S_{t\wedge \tau}, \ t \in [0,T]\}$ with respect to the information flows $\bG$ and $\widetilde \bG$ respectively, and we show that, under suitable conditions, $S^\tau$ satisfies the so-called {\em structure condition} with respect to both $\bG$ and $\widetilde\bG$ on the stochastic interval $\llbracket 0, \tau\wedge T\rrbracket$, see e.g. \citet[Section 1, page 1540]{schweizer1994approximating} for further details. 

The structure condition of the stopped price process is a relevant tool for the computation of the minimal martingale measure and the orthogonal decompositions that allow to characterize the optimal hedging strategy under full and partial information.
Moreover, the semimartingale decomposition of $S^\tau$ with respect to the information flow $\widetilde\bG$ allows to reduce the hedging problem under partial information to a full information problem where all involved processes are $\widetilde\bG$-adapted.

\begin{remark}\label{W}
Recall that if the process  $F$ given in  \eqref{eq:cond_distribution_F} is increasing, for any given $\bF$-predictable  $(\bF, \P)$-martingale, $m =\{ m_t, \ t \in [0,T]\}$,  the stopped process $m^\tau= \{m_{t\wedge \tau}, \ t \in [0,T]\}$ is a $(\bG,\P)$-martingale, see \citet[Lemma 5.1.6]{bielecki2002}.

Since $F$ is increasing in our setting, both processes $W^\tau= \{W_{t\wedge \tau}, \ t \in [0,T]\}$ and  $B^{\tau}= \{B_{t\wedge \tau}, \ t \in [0,T]\}$ are $(\bG,\P)$-martingales over $\llbracket 0, \tau\wedge T\rrbracket$.

Moreover, by L\'evy's Theorem we also obtain that $W^\tau$ and $B^{\tau}$ are $(\bG,\P)$-Brownian motions on $\llbracket 0, \tau\wedge T\rrbracket$
 and, as a consequence, the integral processes $\ds \left\{\int_0^t \varphi_s \ud W^\tau_s, \ t \in [0,T]\right\}$ and $\ds \left\{\int_0^t \varphi_s \ud B^{\tau}_s, \ t \in [0,T]\right\}$ are $(\bG,\P)$-(local) martingales for any $\bG$-predictable process $\varphi=\{\varphi_t, \ t \in [0,T]\}$.
\end{remark}

By Remark \ref{W}, we get that the stopped process $S^\tau$ is a $(\bG, \P)$-semimartingale, decomposable as the sum of a locally square integrable $(\bG, \P)$-local martingale and a $(\bG, \P)$-predictable process of finite variation, both null at zero, i.e.
\begin{align}
S^\tau_t = s_0 + \int_0^{t \wedge \tau} S^\tau_u  \mu(u, S^\tau_u, X^\tau_u) \ud u + \int_0^{t \wedge \tau} S^\tau_u \sigma(u, S^\tau_u) \ud W^\tau_u, \quad t \in [0,T],
\end{align}
where
\begin{align}
 X^\tau_t = x_0 + \int_0^{t \wedge \tau}   b(u, X^\tau_u) \ud u  + \int_0^{t \wedge \tau}  a(u, X^\tau_u)  \left[ \rho \ud W^\tau_u+ \sqrt{1-\rho^2} \ud B^\tau_u\right], \quad t \in [0,T].
 \end{align}

Since  $S^\tau$ is $\widetilde \bG$-adapted, then it also admits a semimartingale decomposition with respect to the information flow $\widetilde \bG$, which will be computed below by means of the (stopped) innovation process $I^\tau$, defined below in \eqref{I}.

Given any subfiltration $\bH=\{\H_t,\ t \in [0,T]\}$ of $\bG$,
we will use the notation  ${}^{o,\bH} {Y}$ (respectively ${}^{p,\bH} Y$) to indicate the optional (respectively predictable) projection of a given $\P$-integrable, $\bG$-adapted process $Y$ with respect to $\bH$ and $\P$, defined as the unique $\bH$-optional (respectively $\bH$-predictable) process such that $ {}^{o,\bH}Y_{\widehat \tau} = \esp{Y_{\widehat \tau} | \H_{\widehat \tau}}$ $\P$-a.s. (respectively $ {}^{p,\bH}Y_{\widehat \tau} = \esp{Y_{\widehat \tau} | \H_{\widehat \tau^-}}$ $\P$-a.s.) on $\{ \widehat \tau < \infty\}$ for every $\bH$-optional (respectively $\bH$-predictable) stopping time $\widehat \tau$.

Moreover, in the sequel we denote by ${}^{o,\widetilde \bG} \mu$, ${}^{p,\widetilde \bG} \mu$, the optional projection and the predictable projection respectively of the process $\{ \mu(t, S^\tau_t, X^\tau_t),  \ t \in [ 0, T] \}$ with respect to the information flow $\widetilde \bG$.


\begin{lemma}\label{lemma:stopped_inovation}
Under Assumption \ref{ipot1}, the process $I^\tau=\{I_t^\tau,\ t \in [0,T]\}$
defined by
\begin{equation}\label{I}
I^\tau_t := W^\tau_t + \int_0^{t\wedge \tau}  \frac{\mu(u, S^\tau_u, X^\tau _u) - {}^{p,\widetilde \bG} \mu_u}{\sigma\left(u,S^\tau_u\right) } \ud u,
\quad t \in [0,T],
\end{equation}
is a $(\widetilde \bG, \P)$-Brownian motion on $\llbracket 0, \tau\wedge T\rrbracket$.
\end{lemma}

The proof is postponed to Appendix \ref{appendix:proofs}.

Lemma \ref{lemma:stopped_inovation} allows to get the following $\widetilde \bG$-semimartingale decomposition of $S^\tau$,
\begin{equation}\label{semimg}
S^\tau_t = s_0 + \int_0^{t \wedge \tau} \!\! S^\tau_u  \hskip 1 mm {}^{p,\widetilde \bG}  \mu_u \ud u + \int_0^{t \wedge \tau} \!\! S^\tau_u \hskip 1mm \sigma(u, S^\tau_u) \ud I^\tau_u, \quad t \in [0,T],
\end{equation}
i.e. the sum of a locally square integrable $(\widetilde \bG, \P)$-local martingale and a $(\widetilde \bG, \P)$-predictable process of finite variation both null at zero.

Moreover, $S^\tau$ satisfies the structure condition with respect to both the filtrations $\bG$ and $\widetilde \bG$. Precisely,
\begin{align}
S^\tau_t &=  s_0 + M^\G_t + \int_0^{t\wedge \tau} \alpha^\G_u \ud \langle M^\G \rangle_u, \quad t \in \llbracket 0, \tau\wedge T\rrbracket,\\
S^\tau_t &= s_0 + M^{\widetilde \G}_t +  \int_0^{t\wedge \tau} \alpha^{\widetilde\G}_u  \ud \langle M^{\widetilde\G} \rangle_u, \quad t \in \llbracket 0, \tau\wedge T\rrbracket,
\end{align}
where $M^\G=\{M_t^\G,\ t \in [0,T]\}$ and $M^{\widetilde\G}=\{M_t^{\widetilde \G},\ t \in [0,T]\}$ are the locally square integrable  $(\bG, \P)$-local martingale and $(\widetilde \bG, \P)$-local martingale respectively, given by
\begin{equation}\label{eq:martG}
M^\G_t :=   \int_0^{t\wedge \tau} S^\tau_u  \sigma(u, S^\tau_u) \ud W^\tau_u, \quad  M^{\widetilde\G}_t :=  \int_0^{t\wedge \tau} S^\tau_u \sigma(u, S^\tau_u) \ud I^\tau_u,\quad t \in [0,T],
\end{equation}
and $\alpha^\G=\{\alpha_t^\G,\ t \in [0,T]\}$ and $\alpha^{\widetilde \G}=\{\alpha_t^{\widetilde \G},\ t \in [0,T]\}$ are the $\bG$-predictable and $\widetilde \bG$-predictable processes, respectively given by
$$ \alpha^\G_t :=  { \mu(t, S^\tau_t, X^\tau_t)  \over S^\tau_t \sigma^2(t, S^\tau_t) }, \quad  \alpha^{\widetilde \G}_t :=  { {}^{p,\widetilde \bG}  \mu_t  \over S^\tau_t \sigma^2(t, S^\tau_t)}, \quad t \in [0,T].
$$

\section{Local risk-minimization for payment streams under partial information} \label{sec:lrm}

The combined financial-insurance market model outlined in Section \ref{sec:market model} is not complete. This frequently occurs in the insurance framework where typically the number of random sources is larger than the number of tradeable risky assets due to the presence of a totally inaccessible death time. Moreover, here additional randomness is brought by the unobservable stochastic factor $X$. This implies that a self-financing hedging strategy in the classical sense does not exist. The goal of the current section is to provide a locally risk-minimizing hedging strategy under restricted information for the payment stream associated to the endowment insurance contract $(\xi, Z, \tau)$, and discuss the relationship with the corresponding optimal hedging strategy under full information.

In the sequel, we define the classes of admissible hedging strategies under full and partial information.
\begin{definition}
The space $\Theta^{\bF, \tau}$ consists of all $\R$-valued $\bF$-predictable processes $\theta=\{\theta_t, \ t \in \llbracket 0, T\wedge \tau\rrbracket\}$ satisfying
\begin{align*} \esp{\int_0^{T \wedge \tau} \left(\theta_u \sigma(u, S^\tau_u) S^\tau_u \right) ^2 \ud u + \left( \int_0^{T \wedge \tau}  |\theta_u  \hskip 1mm \mu(u, S^\tau_u, X^\tau_u) S^\tau_u| \ud u \right)^2  } < \infty.
\end{align*}
\end{definition}
\begin{definition}
The space $\Theta^{\bF^S, \tau}$ consists of all $\R$-valued $\bF^S$-predictable processes $\theta=\{\theta_t, \ t \in \llbracket 0, T\wedge \tau\rrbracket\}$ satisfying
\begin{align*} \esp{\int_0^{T \wedge \tau} \left(\theta_u \sigma(u, S^\tau_u) S^\tau_u \right) ^2 \ud u + \left( \int_0^{T \wedge \tau}  |\theta_u \ {}^{p, \widetilde \bG} \mu_u \ S^\tau_u| \ud u \right)^2 } < \infty.\end{align*}
\end{definition}

\begin{remark}
Notice that for $\theta \in \Theta^{\bF, \tau}$ (respectively  $\theta \in \Theta^{\bF^S, \tau}$), we get
\begin{itemize}
\item[(i)] $\ds \int_0^{t\wedge \tau} \theta_u \ud S_u=\int_0^{t} \theta_u \ud S^\tau_u$, for every $t \in [0,T]$, see \citet[Chapter VIII, equation 3.3]{dellacherie_meyer1982};
\item[(ii)]the integral process $\ds \left\{\int_0^{t} \theta_u \ud S^\tau_u, \ t \in [0, T] \right\}$  is a $(\bG, \P)$-semimartingale (respectively $(\widetilde \bG, \P)$-semimartingale), see \citet[Chapter 3.II]{shiryaev1998}.
\end{itemize}
\end{remark}

\begin{definition}
An $(\bF,\bG)$-strategy (respectively $(\bF^S,\widetilde \bG)$-strategy) is a bidimensional process $\varphi =(\theta, \eta)$  where  $\theta \in \Theta^{\bF, \tau}$ (respectively $\theta \in \Theta^{\bF^S, \tau}$) and
$\eta$ is a real-valued $\bG$-adapted (respectively $\widetilde \bG$-adapted) process such that the associated value process
$V(\varphi) := \theta S^\tau + \eta$ is right-continuous and square integrable over $\llbracket 0, T\wedge \tau \rrbracket$.
\end{definition}

Note that the first component $\theta$ of the $(\bF,\bG)$-strategy (respectively $(\bF^S,\widetilde \bG)$-strategy), which represents the number of risky assets in the portfolio,  is $\bF$-predictable (respectively $\bF^S$-predictable), while the amount $\eta$ invested in the risk-free asset is $\bG$-adapted (respectively $\widetilde \bG$-adapted). This reflects the natural situation where a trader invests in the risky asset according to her/his knowledge on the asset prices before the death of the policyholder and rebalances the portfolio also upon the death information.

Following \citet{schweizer2008local}, we assign to each admissible strategy a cost process.
\begin{definition}
The {\em cost process} $C(\varphi)$ of an $(\bF,\bG)$-strategy (respectively $(\bF^S,\widetilde \bG)$-strategy) $\varphi=(\theta, \eta)$ is given by
\begin{align}\label{eq:cost}C_t(\varphi) := N_t + V_t(\varphi) - \int_0^t \theta_u \ud S^\tau_u, \quad t \in \llbracket 0, T\wedge \tau \rrbracket,
\end{align}
where $N$ is defined in \eqref{N}.

A $(\bF, \bG)$-strategy (respectively $(\bF^S,\widetilde \bG)$-strategy) $\varphi$ is called {\em mean-self-financing} if its cost process $C(\varphi)$ is a $(\bG, \P)$-martingale (respectively $(\widetilde \bG, \P)$-martingale).
\end{definition}

It is well known in the literature (see e.g. \citet{moller2001risk,schweizer2008local, biagini2012local}) that a natural extension of the local risk-minimization approach to payment streams requires to look for admissible strategies satisfying the $0$-achieving property, that is,
$$V_{\tau \wedge T}(\varphi)=0, \quad \P-\mbox{a.s.}.$$
Then, by \citet[Theorem 1.6]{schweizer2008local}, we provide the following equivalent definition of locally risk-minimizing strategy.
\begin{definition}
Let $N$ be the payment stream given in \eqref{N} associated to the endowment insurance contract $(\xi, Z, \tau)$.
We say that an $(\bF,\bG)$-strategy (respectively $(\bF^S,\widetilde \bG)$-strategy) $\varphi$ is $(\bF,\bG)$-{\em locally risk-minimizing} (respectively $(\bF^S,\widetilde \bG)$-{\em locally risk-minimizing}) for $N$ if
\begin{itemize}
\item[(i)]$\varphi$ is $0$-achieving and mean-self-financing,
\item[(ii)]the cost process  $C(\varphi)$, defined in \eqref{eq:cost}, is strongly orthogonal to the $\bG$-martingale part $M^\G$ (respectively $\widetilde \bG$-martingale part $M^{\widetilde \G}$) of $S^\tau$, both given in \eqref{eq:martG}.
\end{itemize}
\end{definition}

Locally risk-minimizing hedging strategies can be characterized via the F\"{o}llmer-Schweizer decomposition  of payment streams associated to life insurance contracts under partial information.

We recall the definition of stopped F\"{o}llmer-Schweizer  decomposition of a square integrable random variable with respect to $\bG$ and $\widetilde \bG$.
\begin{definition}[Stopped F\"{o}llmer-Schweizer decomposition with respect to $\bG$]
Given a random variable $\zeta\in L^2(\G_T, \P)$, we say that $\zeta$ admits a {\em stopped F\"{o}llmer-Schweizer decomposition with respect to $\bG$}, if there exist a process $\theta^\F \in \Theta^{\bF, \tau}$, a square integrable $(\bG, \P)$-martingale $A^\G=\{A^\G_t, \ t \in \llbracket0,T\wedge\tau\rrbracket \}$ null at zero, strongly orthogonal to the martingale part of $S^\tau$, $M^\G$, given in \eqref{eq:martG}, and $\zeta_0\in \R$ such that
\begin{equation}\label{eq:FS-G}
\zeta=\zeta_0+\int_0^T\theta^\F_u \ud S^\tau_u+A^\G_{T\wedge\tau}, \quad \P-a.s..
\end{equation}
\end{definition}

\begin{definition}[Stopped F\"{o}llmer-Schweizer decomposition with respect to $\widetilde \bG$]
Given a random variable $\zeta\in L^2(\widetilde \G_T, \P)$, we say that $\zeta$ admits a {\em stopped F\"{o}llmer-Schweizer decomposition with respect to $\widetilde \bG$}, if there exist a process $\theta^{\F^S} \in \Theta^{\bF^S, \tau}$, a square integrable $(\widetilde \bG, \P)$-martingale $A^{\widetilde \G}=\{A^{\widetilde \G}_t, \ t \in \llbracket  0,T\wedge\tau \rrbracket \}$ null at zero, strongly orthogonal to the martingale part of $S^\tau$,$M^{\widetilde \G}$, given in \eqref{eq:martG}, and $\zeta_0\in \R$ such that
\begin{equation}\label{eq:FS-Gtilde}
\zeta=\zeta_0+\int_0^T\theta^{\F^S}_u \ud S^\tau_u+A^{\widetilde \G}_{T\wedge\tau}, \quad \P-a.s..
\end{equation}
\end{definition}

Under Assumption \ref{ipot1}, the {\em mean variance tradeoff processes} $K=\{K_t,\ t \in [0,T]\}$ and $\widetilde K=\{\widetilde K_t,\ t \in [0,T]\}$ under $\bG$ and $\widetilde \bG$,  respectively defined by
\begin{align}
K_t:=\int_0^{t} (\alpha^\G_u)^2\ud \langle M^\G \rangle_u, \qquad \widetilde K_t:=\int_0^{t} (\alpha^{\widetilde\G}_u)^2\ud \langle M^\G \rangle_u, \quad t \in [0,T], \label{eq:mvt_Gtilde}
\end{align}
are bounded uniformly in $t$ and $\omega$. This guarantees the existence of decompositions \eqref{eq:FS-G} and \eqref{eq:FS-Gtilde} for every square integrable random variable, see e.g. \citet[Section 5]{schweizer1994approximating} and references therein. Other classes of sufficient conditions  for the existence of the F\"{o}llmer-Schweizer decompositions can be found e.g. in \citet{schweizer1995minimal}, \citet{monat1995follmer}, \citet{choulli2010follmer} and \citet{ceci2014bsdes}.

The following proposition gives a characterization of the optimal hedging strategy.
\begin{proposition}\label{prop:risk_min_strategy}
 Let  $N$ be the payment stream associated to the endowment insurance contract $(\xi, Z, \tau)$. Then, $N$ admits an $(\bF^S, \widetilde \bG)$-locally risk-minimizing strategy $\varphi^*=(\theta^*, \eta^*)$ if and only if $N_{T\wedge\tau}=\xi \I_{\{\tau>T\}}+Z_\tau\I_{\{\tau\leq T\}}$ admits a stopped F\"{o}llmer-Schweizer decomposition with respect to $\widetilde \bG$, i.e.
 \begin{equation}\label{eq:decomp_1}
 N_{T\wedge \tau}=\zeta_0+\int_0^T\theta^{\F^S}_u \ud S^\tau_u + A^{\widetilde \G}_{t\wedge\tau} \quad \P-a.s..
 \end{equation}
 Finally, the strategy $\varphi^*$ is explicitly given by
\begin{align}\label{eq:strategy}
\theta^*=\theta^{\F^S}, \quad \eta^*=V(\varphi^*)-\theta^{\F^S}S^\tau,
\end{align}
with value process
\begin{align}\label{eq:port_value}
V_t(\varphi^*)=\zeta_0+\int_0^t\theta^{\F^S}_u \ud S^\tau_u+A^{\widetilde \G}_t-N_t, \quad t \in \llbracket 0,T\wedge\tau\rrbracket,
\end{align}
and minimal cost
\begin{align}\label{eq:opt_cost}
C_t(\varphi^*)=\zeta_0+A^{\widetilde \G}_t, \quad t \in \llbracket 0,T\wedge\tau\rrbracket.
\end{align}
\end{proposition}

\begin{proof}
The proof follows by that of \citet[Proposition 3.7]{biagini2012local}, by replacing the filtrations $\bG$ and $\bF$ with $\widetilde \bG$ and $\bF^S$, respectively.

Precisely, if $ N_{T\wedge \tau}$ has the stopped F\"{o}llmer-Schweizer decomposition with respect to $\widetilde \bG$ \eqref{eq:decomp_1}, then \eqref{eq:strategy} and \eqref{eq:port_value} define an $(\bF^S, \widetilde \bG)$-strategy with cost \eqref{eq:opt_cost}. It is easy to see that $C(\varphi^*)$ is a martingale and that $\varphi^*$ is $0$-achieving, and therefore $\varphi^*$ is an $(\bF^S, \widetilde \bG)$-locally risk-minimizing strategy. For the converse implication, note that if $\varphi$ is $(\bF^S, \widetilde\bG)$-locally risk minimizing, then it is $0$-achieving and mean-self-financing, and we get
\[
N_{T\wedge \tau}=C_{T\wedge\tau}(\varphi)+\int_0^T\theta_u \ud S^\tau_u=C_{0 }(\varphi)+\int_0^T\theta_u \ud S^\tau_u+ (C_{T\wedge\tau}(\varphi)-C_{0}(\varphi)),
\]
which is equivalent to \eqref{eq:decomp_1} with $\ds \zeta_0:=C_{0}(\varphi), \quad \theta^{\F^S}:=\theta, \quad A^{\widetilde \G}:=C(\varphi)-C_{0}(\varphi)$. Finally note that $A^{\widetilde \G}$ is strongly orthogonal to the $\widetilde \bG$-martingale part of $S^\tau$, which concludes the proof.
\end{proof}

\subsection{The optimal strategy via the Galtchouk-Kunita-Watanabe decomposition}

When $S^\tau$ is continuous and satisfies the structure condition, the (stopped) F\"{o}llmer-Schweizer decomposition of a given square integrable random variable with respect to $S^\tau$ can be computed by switching to the minimal martingale measure. In the following, we provide the definition of the minimal martingale measure adapted to 
the combined financial-insurance market model outlined in Section \ref{sec:market model}.

\begin{definition}\label{def:mmm}
A martingale measure $\widehat \P$ equivalent to $\P$ with square integrable density is called {\em minimal} for $S^\tau$ if any square integrable $(\bG, \P)$-martingale, which is strongly orthogonal to the martingale part of $S^\tau$, $M^\G$ given in \eqref{eq:martG}, under $\P$ is also a $(\bG, \widehat \P)$-martingale.
\end{definition}

Define the process $L^\tau=\{L^\tau_t, \ t \in [0,T]\}$ by setting
\begin{equation}\label{eq:Ltau}
L^\tau_t = \left.\frac{ \ud \widehat\P }{ \ud \P}\right|_{\G_{\tau\wedge t}} := \mathcal E\left( - \int_0^\cdot  {\mu(u, S^\tau_u, X^\tau_u)  \over \sigma(u,S^\tau_u) } \ud  W^\tau_u  \right)_{t\wedge \tau}, \quad t \in [0,T]. \end{equation}
By condition (ii) of Assumption \ref{ipot1}, we get that $L^\tau_t \in L^2(\G_t, \P)$, for every $t \in [0,T]$.

Applying the results in \citet{ansel1993unicite}, we get that $\widehat \P$, given in \eqref{eq:Ltau}, corresponds to the minimal martingale measure.
By the Girsanov theorem the process $\widehat W ^\tau=\{\widehat W ^\tau_t, \ t \in [0,T]\}$, defined by
$$ \widehat W ^\tau_t := W^\tau _t + \int_0^{t \wedge \tau}  {\mu(u, S^\tau_u, X^\tau_u)  \over \sigma(u,S^\tau_u) } \ud u, \quad t \in [0, T],$$
is a $(\bG, \widehat \P)$-Brownian motion.

Note that $L^\tau$ and $\widehat W^\tau$ coincide with the processes $L$ and $\widehat W$, given in \eqref{eq:densitaL} and \eqref{eq:What} respectively, on the stochastic interval $\llbracket 0, T\wedge\tau\rrbracket$.
\begin{remark}
We may also define the minimal martingale measure $\widehat \Q$ for $S^\tau$ with respect to the information flow $\widetilde \bG$, by setting
$$
\left.{ \ud \widehat \Q \over \ud \P}\right|_{\widetilde \G_{\tau\wedge T}} := \doleans{ - \int_0^\cdot \alpha^{\widetilde \G}_u \ud  M^{\widetilde \G }_u }_{T\wedge \tau} =  \doleans{- \int_0^\cdot  {{}^{p,\widetilde \bG} \mu_u  \over \sigma(u,S^\tau_u) }  \ud  I^\tau_u}_{T\wedge \tau}.$$
Since $S^\tau$ has continuous trajectories, $ \widehat \Q$ coincides with the restriction of $\widehat \P$ over $\widetilde \G_{\tau\wedge T}$, see, e.g. \citet[Lemma 4.3]{ceci2015local}. Indeed, by \eqref{I}
 $$I^\tau_t + \int_0^{t \wedge \tau}  {{}^{p,\widetilde \bG}  \mu_u  \over \sigma(u,S^\tau_u) } \ud u =
W^\tau _t + \int_0^{t \wedge \tau}  {\mu(u, S^\tau_u, X^\tau_u)  \over \sigma(u,S^\tau_u) } \ud u =
\widehat W ^\tau_t, \quad t \in [0,T],$$
which, therefore, implies that the process $\ds \left\{I^\tau_t + \int_0^{t \wedge \tau}  {{}^{p,\widetilde \bG}  \mu_u  \over \sigma(u,S^\tau_u) } \ud u,\ t \in [0,T] \right\}$ is a
$(\widetilde \bG, \widehat \P)$-Brownian motion since it is $\widetilde \bG$-adapted.
\end{remark}
In the following we show that the F\"{o}llmer-Schweizer decomposition of the payment stream $N$ associated to the endowment insurance contract $(\xi, Z, \tau)$ indeed coincides with its Galtchouk-Kunita-Watanabe decomposition under the minimal martingale measure, which is easier to characterize.

For reader's convenience, we recall the definition of the Galtchouk-Kunita-Watanabe decomposition of an $\R$-valued local martingale, adapted to this setting.
\begin{definition}[Galtchouk-Kunita-Watanabe decomposition]
Any $\R$-valued $\bG$-local martingale (respectively $\widehat \bG$-local martingale) $\zeta=\{\zeta_t,\ t \in \llbracket 0, T \wedge \tau \rrbracket\}$ admits a {\em Galtchouk-Kunita-Watanabe decomposition} with respect to $S^\tau$ under $\widehat \P$, that is, it can be uniquely written as
\begin{equation}
\zeta_t=\zeta_0 + \int_0^t \bar \theta_u \ud S_u^\tau + \bar A_t, \quad \widehat \P-a.s., \quad t \in \llbracket0,T\wedge \tau\rrbracket,
\end{equation}
where $\zeta_0 \in \R$, $\bar \theta=\{\bar \theta_t, \ t \in \llbracket 0,T\wedge\tau\rrbracket\}$  is a $\bG$-predictable (respectively $\widehat \bG$-predictable) process
such that $\{\int_0^{t}\left(\bar \theta_u \sigma(u,S_u^\tau)S_u^\tau)\right)^2\ud u,\ t \in \llbracket0,T\wedge \tau\rrbracket\}$ is locally $\widehat \P$-integrable and $\bar A=\{\bar A_t,\ t \in \llbracket0,T\wedge \tau\rrbracket\}$ is a $(\bG,\widehat \P)$-local martingale (respectively $(\widehat \bG,\widehat \P)$-local martingale) null at zero, strongly orthogonal to $S^\tau$.
\end{definition}

Consider the payment stream $N$ associated to the endowment insurance contract $(\xi, Z, \tau)$, and define the process $\widehat V=\{\widehat V_t,\ t \in \llbracket0,T\wedge \tau\rrbracket\}$ by setting
\begin{equation}\label{eq:hatV}
\widehat V_t:=\widehatesp{N_{T\wedge\tau}|\widetilde \G_t}, \ t \in \llbracket 0,T\wedge \tau\rrbracket,
\end{equation}
where $\widehatesp{Y \Big|\widetilde \G_t}$ denotes the conditional expectation of a $\widehat \P$-integrable random variable $Y$ with respect to $\widehat \P$ and the $\sigma$-algebra $\widetilde \G_t$, for every $t \in [0,T]$.

Since $S^\tau$ is a $(\widetilde \bG, \widehat \P)$-martingale and $\widehat V_t\in L^1(\widetilde \G_t, \widehat \P)$, for every $t \in \llbracket0,T \wedge \tau\rrbracket $, then
 $\widehat V$ admits the Galtchouk-Kunita-Watanabe decomposition with respect to $S^\tau$ under $(\widetilde \bG, \widehat \P)$ given by
\begin{equation}\label{eq:GKW1}
\widehat V_t= \espp{N_{T\wedge\tau}}+\int_0^t \widetilde \theta_u \ud S^\tau_u+ \widehat A_{t},  \quad t \in \llbracket0,T \wedge \tau\rrbracket,
\end{equation}
where $\widetilde \theta=\{\widetilde \theta_t, \ t \in \llbracket 0,T\wedge\tau\rrbracket\}$ is a $ \widetilde \bG$-predictable, integrable process with respect to $S^\tau$, $\widehat A=\{\widehat A_t, \ t \in \llbracket 0,T\wedge\tau\rrbracket\}$ is a $(\widetilde \bG,\widehat \P)$-martingale null at time zero, strongly orthogonal to $S^\tau$.
It is always possible to replace $\widetilde \theta$ by an $\bF^S$-predictable process $\widehat \theta$ such that $\I_{\{\tau\geq t\}}\widetilde \theta_t=\I_{\{\tau\geq t\}}\widehat \theta_t$, for each $ t \in [0,T]$. Then, equation \eqref{eq:GKW1} can be written as
\begin{equation}\label{eq:GKW}
\widehat V_t= \espp{N_{T\wedge\tau}}+\int_0^t \widehat \theta_u \ud S^\tau_u+ \widehat A_{t}, \quad t \in \llbracket0,T \wedge \tau\rrbracket.
\end{equation}

\begin{theorem}\label{thm:equivalence}
Let $N$ be the payment stream given by \eqref{N}, associated to the the endowment insurance contract $(\xi, Z, \tau)$. If either $N_{T\wedge\tau}$ admits a stopped F\"{o}llmer-Schweizer decomposition with respect to $\widetilde \bG$, or
$\widehat \theta\in \Theta^{\F^S,\tau}$ and $\widehat A$ is a square integrable $(\widetilde \bG,\P)$-martingale null at time zero, strongly orthogonal to the martingale part of $S^\tau$, $M^{\widetilde \G}$, then \eqref{eq:GKW} for $t=T\wedge\tau$ gives the stopped F\"{o}llmer-Schweizer decomposition of $N_{T\wedge\tau}$ with respect to $\widetilde \bG$.
\end{theorem}

\begin{proof}
The proof follows by that of \citet[Theorem 3.9]{biagini2012local} or \citet[Theorem 3.5]{schweizer2001guided}, by replacing the filtrations $\bG$ and $\bF$ with $\widetilde \bG$
and $\bF^S$, respectively.
\end{proof}

The above theorem provides the existence of the optimal hedging strategy. Now, we need a more explicit characterization which allows to compute the locally risk-minimizing strategy for any kind of unit-linked life insurance contract of the form $(\xi, Z, \tau)$ given in Definition \ref{def:dafaultable_claim}. In the sequel, given any subfiltration $\bH$ of $\bG$, the notation ${}^{\widehat p, \bH}Y$ refers to the $(\bH, \widehat \P)$-predictable projection of a given $\widehat \P$-integrable $\bG$-adapted process $Y$.

Proposition \ref{thm:strategy0} below provides a representation of the integrand in the Galtchouk-Kunita-Watanabe decomposition of $N_{T\wedge \tau}$ under partial information in terms of the corresponding Galtchouk-Kunita-Watanabe decomposition under full information, and then Theorem \ref{thm:strategy} will give the characterization of the locally risk-minimizing strategy for the insurance claim $(\xi, Z, \tau)$ under partial information.

\begin{proposition}\label{thm:strategy0}
Let $N$ be the payment stream given by \eqref{N}, associated to the endowment insurance contract $(\xi, Z, \tau)$ and assume $N_{T\wedge \tau}$ and $S^\tau$ to be $\widehat \P$-square integrable.

Consider the Galtchouk-Kunita-Watanabe decomposition of $N_{T\wedge \tau}$ with respect to $(\bG, \widehat \P)$, i.e.
\begin{equation} \label{g1} N_{T\wedge \tau} = \espp{N_{T\wedge\tau}}+\int_0^T\widehat \theta^{\F}_u \ud S^\tau_u+\widehat A^{\G}_{T\wedge\tau}, \quad \widehat \P-a.s.,  \end{equation}
where $\widehat \theta^{\F}$ is an $\bF$-predictable  process such that $\widehatesp{ \int_0^{T\wedge\tau}\left( \widehat \theta^{\F}_u \sigma(u, S^\tau_u) S^\tau_u \right)^2 \ud u} < \infty$ and $\widehat A^{\G}=\{\widehat A_t^{\G},\ t \in \llbracket0, T\wedge\tau\rrbracket\}$ is a $(\bG, \widehat \P)$-martingale null at zero strongly orthogonal to $S^\tau$.

Then, $N_{T\wedge \tau}$ has the following Galtchouk-Kunita-Watanabe decomposition with respect to $(\widetilde \bG, \widehat \P)$
\begin{equation} \label{g2}N_{T\wedge \tau} = \espp{N_{T\wedge\tau}}+\int_0^T\widehat \theta^{\F^S}_u \ud S^\tau_u+\widehat A^{\widetilde \G}_{T\wedge\tau}, \quad \widehat \P-a.s.,
\end{equation}
where
\begin{equation}\label{thetaFS}
\widehat \theta^{\F^S}_t =  {  {}^{\widehat p, \bF^S} ( \widehat \theta^\F_t  e^{- \int_0^t \gamma_u \ud u})  \over {}^{\widehat p, \bF^S} ( e^{- \int_0^t \gamma_u \ud u} ) },  \quad t \in \llbracket0, T\wedge\tau\rrbracket,
\end{equation}
and the $(\widetilde \bG, \widehat \P)$-martingale $\widehat A^{\widetilde \G}=\{\widehat A_t^{\widetilde\G},\ t \in \llbracket0, T\wedge\tau\rrbracket\}$ is given by
 \begin{equation} \label{mgGtilde}
\widehat A^{\widetilde \G}_t= \widehatesp{\widehat A^{\G}_t \Big{|}\widetilde \G_t} + \widehatesp{\int_0^t ( \widehat \theta^{\F}_u  -  \widehat \theta^{\F^S}_u)  \ud S^\tau_u \Big{|}\widetilde \G_t},
\quad t \in \llbracket0, T\wedge\tau\rrbracket. \end{equation}

\end{proposition}

\begin{proof}
In virtue of Corollary \ref{cor:pred_proj}, if $\widehat \theta^{\F^S}$ satisfies \eqref{thetaFS}, then
$$
\widehat \theta^{\F^S}_t = {}^{\widehat p,\widetilde \bG} \widehat \theta^{\F}_t, \quad t \in \llbracket 0,T\wedge\tau\rrbracket.
$$

By decomposition \eqref{g1}  we can write
\begin{equation} \label{gkw_w}
N_{T\wedge \tau} = \espp{N_{T\wedge\tau}}+\int_0^T  {}^{\widehat p,\widetilde \bG} \widehat  \theta^{\F}_u \ud S^\tau_u+ \widetilde A_{T\wedge\tau} + \widehat A^\G_{T\wedge\tau}, \quad \widehat \P-\mbox{a.s.}, \end{equation}
where $\widetilde A=\{\widetilde A_t,\ t \in \llbracket 0,T\wedge\tau\rrbracket\}$, given by
$$
\widetilde A_t := \int_0^t (\widehat \theta^{\F}_u  -  \widehat \theta^{\F^S}_u)  \ud S^\tau_u =
\int_0^t (\widehat \theta^{\F}_u  -  {}^{\widehat p,\widetilde \bG} \widehat \theta^{\F}_u)  \ud S^\tau_u , \quad t \in \llbracket 0,T\wedge\tau\rrbracket,
$$
is a square integrable $(\bG, \widehat\P)$-martingale. This is a consequence of the fact that $S^\tau$ is a $(\bG, \widehat \P)$-martingale and that, by Jensen's inequality the following holds
\begin{align*}
& \widehatesp{ \int_0^{T\wedge\tau}\!\!\!\left( \widehat \theta^{\F^S}_u \sigma(u, S^\tau_u) S^\tau_u\right)^2 \ud u}  = \widehatesp{ \int_0^T
 \left({}^{\widehat p,\widetilde \bG} \widehat \theta^{\F}_u \sigma(u, S^\tau_u) S^\tau_u\right)^2  \I_{\{\tau \geq u\}} \ud u} \\
& \qquad \leq \widehatesp{ \int_0^T
 {}^{\widehat p,\widetilde \bG}\left(\left(\widehat \theta^{\F}_u \sigma(u, S^\tau_u) S^\tau_u\right)^2  \I_{\{\tau \geq u\}}\right) \ud u}
 = \widehatesp{ \int_0^{T\wedge\tau}\!\!\left(\widehat \theta^{\F}_u \sigma(u, S^\tau_u) S^\tau_u\right)^2 \ud u} < \infty.
\end{align*}
By \eqref{mgGtilde}, conditioning  \eqref{gkw_w} with respect to $\widetilde \G_{T\wedge\tau}$  yields
\begin{align}
N_{T\wedge \tau} &=  \espp{N_{T\wedge\tau}} +\int_0^T  {}^{\widehat p,\widetilde \bG} \widehat  \theta^{\F}_u \ud S^\tau_u+ \widehatesp{ \widetilde A_{T\wedge\tau} + \widehat A^\G_{T\wedge\tau} | \widetilde \G_{T\wedge\tau}}\\
& = \espp{N_{T\wedge\tau}}  +  \int_0^T  {}^{\widehat p,\widetilde \bG}  \theta^{\F}_u \ud S^\tau_u+ \widehat A^{\widetilde \G}_{T\wedge\tau}.
\end{align}
This provides the Galtchouk-Kunita-Watanabe decomposition of  $N_{T\wedge \tau}$ with respect to $(\widetilde \bG, \widehat \P)$, once we verify that the square integrable $(\widetilde \bG, \widehat \P)$-martingale $\widehat A^{\widetilde \G}$  is strongly orthogonal to $S^\tau$.
Note that $\widehat A^{\widetilde \G}$ satisfies
$$
\widehatesp{ \widehat A^{\widetilde \G}_{T\wedge\tau} \int_0^{T\wedge\tau} \varphi_u  \ud S^\tau_u} =0,
$$
for all $\widetilde \bG$-predictable processes $ \varphi$ such that  $\widehatesp{ \int_0^{T\wedge\tau} \varphi^2_u\   \ud \langle   S^\tau \rangle_u} < \infty$, i.e. $\widehat A^{\widetilde \G}$ is $\widetilde \bG$-weakly orthogonal to $S^\tau$, see Definition 2.1 in \citet{ceci2014gkw}. Indeed, since $\varphi$ is $\widetilde \bG$-predictable, by the tower rule
$$
\widehatesp{ \widehat A^{\widetilde \G}_{T\wedge\tau} \int_0^{T\wedge\tau} \varphi_u  \ud S^\tau_u} =\widehatesp{ \widehat A^{\G}_{T\wedge\tau} \int_0^{T\wedge\tau} \varphi_u  \ud S^\tau_u}+\widehatesp{ \widetilde A_{T\wedge\tau} \int_0^{T\wedge\tau} \varphi_u  \ud S^\tau_u}.
$$
Both of the terms on the right-hand side are zero: the first one because $\widehat A^\G$ is strongly orthogonal to $S^\tau$, and the second one follows by the computations below,
\begin{align*}
& \widehatesp{ \widetilde A_{T\wedge\tau} \int_0^{T\wedge\tau} \! \!\!  \varphi_u \ud S^\tau_u}  = \widehatesp{  \int_0^{T\wedge\tau}\! \!  \varphi_u ( \widehat \theta^{\F}_u  -  {}^{\widehat p,\widetilde \bG} \widehat \theta^{\F}_u) \ud \langle  S^\tau \rangle_u}\\
 & \qquad =
\widehatesp{  \int_0^{T\wedge\tau} \! \!  \varphi_u ( \widehat \theta^{\F}_u  -  {}^{\widehat p,\widetilde \bG} \widehat \theta^{\F}_u) \sigma^2(u, S^\tau_u) (S^\tau_u)^2 \ud u} =0,
\end{align*}
since $\{\sigma(t, S^\tau_t) S^\tau_t, \ t \in  \llbracket 0,T\wedge\tau\rrbracket\} $ has continuous trajectories.

Finally, the strong orthogonality between  $ \widehat A^{\widetilde \G}$ and $S^\tau$ follows by the $\widetilde \bG$-weak orthogonality since $ \widehat A^{\widetilde \G}$ is $\widetilde \bG$-adapted (see \citet[Remark 2.4]{ceci2014gkw}).
\end{proof}

\begin{theorem}\label{thm:strategy}
Let $N$ be the payment stream given by \eqref{N}, associated to the endowment insurance contract $(\xi, Z, \tau)$ and let condition (ii) of Assumption \ref{ipot1}  hold.

Then, $N_{T\wedge \tau}$ admits a stopped F\"{o}llmer-Schweizer decomposition with respect to both $\widetilde \bG$ and $\bG$, i.e.
\begin{align}
N_{T\wedge \tau} =& \espp{N_{T\wedge\tau}}+\int_0^T\theta^{\F^S}_u \ud S^\tau_u+A^{\widetilde \G}_{T\wedge\tau}, \quad \P-a.s.,\\
N_{T\wedge \tau} =& \espp{N_{T\wedge\tau}}+\int_0^T\theta^{\F}_u \ud S^\tau_u+A^{\G}_{T\wedge\tau}, \quad \P-a.s.,
\end{align}

with $\theta^{\F^S} = \widehat \theta^{\F^S}$, $A^{\widetilde \G} = \widehat A^{\widetilde \G}$, $\theta^{\F} = \widehat \theta^{\F}$ and $A^{\G}= \widehat A^{\G}$ given in
decompositions \eqref{g1} and \eqref{g2}.
If in addition, $N_{T\wedge\tau}$ and $S^\tau$ are $\widehat \P$-square integrable, then the $(\bF^S, \widetilde \bG)$-locally risk-minimizing strategy $\varphi^*=(\theta^*, \eta^*)$ for $N$ is given by
\begin{align}
\theta^*_t & =  \theta^{\F^S}_t =  \frac{{}^{\widehat p, \bF^S} \left(\theta^\F_t  e^{- \int_0^t \gamma_u \ud u}\right)}{{}^{\widehat p, \bF^S}\left(e^{- \int_0^t \gamma_u \ud u}\right)},  \quad t \in \llbracket0, T\wedge\tau\rrbracket \label{eq:theta*}, \\
\ds \eta^*_t & = V_t(\varphi^*)-\theta^*_t S^\tau_t, \quad t \in \llbracket0, T\wedge\tau\rrbracket, \label{eq:eta*}
\end{align}
and the optimal value process $V(\varphi^*)$ is given by
\begin{align}\label{V^S}
V_t(\varphi^*) =  \espp{N_{T\wedge\tau}}+\int_0^t\theta^*_u \ud S^\tau_u+A^{\widetilde \G}_t-N_t, \quad t \in \llbracket 0,T\wedge\tau\rrbracket,
\end{align}
with
$$
A^{\widetilde \G}_ {t}= \espp{A^{\G}_ {t } \Big{|} \widetilde \G_{t} } + \espp{ \int_0^t \theta^\F_u \ud S^\tau_u  \Big{|} \widetilde \G_t } - \int_0^t\theta^*_u \ud S^\tau_u, \quad t \in \llbracket 0,T\wedge\tau\rrbracket.
$$
\end{theorem}

\begin{proof}
Under condition (ii) of Assumption  \ref{ipot1}, we obtain the existence of the F\"{o}llmer-Schweizer decompositions of $N_{T\wedge\tau}$ with respect to $\bG$ and $\widetilde \bG$. By Theorem  \ref{thm:equivalence}  and Proposition \ref{thm:strategy0} we get that \eqref{g1} and \eqref{g2} give the F\"{o}llmer-Schweizer decompositions of $N_{T\wedge\tau}$ with respect to $\bG$ and $\widetilde \bG$ respectively. Finally, the result follows by Proposition \ref{prop:risk_min_strategy}.
\end{proof}

Representation \eqref{eq:theta*} requires the knowledge of the process $\theta^\F$, that is, the first component of the {\em $(\bF, \bG)$-locally risk-minimizing strategy}.

To characterize the process $\theta^\F$, define  the process $\widehat V^\G=\{\widehat V^\G_t,\ t \in \llbracket0,T\wedge \tau\rrbracket\}$ by setting
\begin{equation}\label{eq:hatVG}
\widehat V^\G_t:=\widehatesp{N_{T\wedge\tau}| \G_t}, \ t \in \llbracket 0,T\wedge \tau\rrbracket.
\end{equation}
Then by \eqref{g1} the process $\widehat  V^\G$ admits the Galchouk-Kunita-Watanabe decomposition given by
\begin{equation}\label{eq:GKWbis}
\widehat V^\G_t= \espp{N_{T\wedge\tau}}+\int_0^t  \theta^\F_u \ud S^\tau_u+ A^\G_{t},  \quad t \in \llbracket0,T \wedge \tau\rrbracket,
\end{equation}
where $A^\G = \widehat A^\G$ is a square integrable $( \bG,\widehat \P)$-martingale null at time zero, strongly orthogonal to $S^\tau$ w.r.t. $\widehat \P$. By taking the predictable covariation with respect to $S^\tau$ on both sides of the equality we get that
\begin{equation}\label{eq:thetaF}
\theta^\F_t = \frac{\ud \langle \widehat V^\G, S^\tau \rangle_t^{\widehat \P}}{ \ud \langle  S^\tau \rangle_t ^{\widehat \P}} \quad t \in \llbracket0,T \wedge \tau\rrbracket,
\end{equation}
where $\langle  \cdot, \cdot  \rangle^{\widehat \P}$ denotes the predictable covariation process under minimal martingale measure   $\widehat \P$.

Now we have to face the task of computing the process $\langle \widehat  V^\G, S^\tau \rangle^{\widehat \P}$.

In the following section we will analyze some examples in a Markovian setting where we are able to give explicit representations of both the hedging strategies $\theta^\F$ and $\theta^{\F^S}$ under full and partial information.

\section{An application: the $\bF$-mortality rate depending on the unobservable stochastic factor} \label{sec:markov}

To introduce a Markovian setting,
we assume that the $\bF$-mortality rate $\gamma$ is of the form $\gamma_t = \gamma(t,X_t)$, $t \in [0,T]$, for a nonnegative measurable function $\gamma$ such that $\esp{\int_0^T\gamma(s, X_s)\ud s}<\infty$, and the endowment insurance contract is given by
the triplet
$(\xi, Z, \tau)$, where $\xi = G(T,S_T)$ and  $Z_t = U(t,S_t)$, for some measurable functions $G$ and $U$ such that  $\bE[|G(T,S_T)|^2 ] < \infty$ and $\bE[ |U(t,S_t)|^2 ] < \infty$,  for every $t \in[0,T]$.

On the probability space $(\Omega, \F, \widehat \P)$ the pair $(S,X)$ satisfies the following system of stochastic differential equations
\begin{equation}
\left\{
\begin{split}
&\ud S_t=S_t \sigma(t,S_t) \ud \widehat W_t, \quad S_0=s_0>0\\
&\ud X_t=\left(b(t,X_t)-a(t,X_t)\rho \ \frac{\mu(t, S_t, X_t)}{\sigma(t,S_t)}\right)\ud t + a(t,X_t)\left(\rho \ud \widehat W_t + \sqrt{1-\rho^2}\ud B_t\right), \quad X_0=x_0\in \R.
\end{split}\right.
\end{equation}

We assume throughout the section that
\begin{equation}\label{hp:generatore_cont}
\widehatesp{\int_0^T\left\{|b(t,X_t)|+ a^2(t,X_t)  + \sigma^2(t,S_t)\right\}\ud t}<\infty.
\end{equation}

The Markovianity of the pair $(S,X)$ under $\widehat \P$ is shown in the Lemma below.
\begin{lemma}\label{markov}
Under conditions (ii) of Assumption \ref{ipot1} and \eqref{hp:generatore_cont}, the pair $(S,X)$ is an $(\bF, \widehat \P)$-Markov process with generator $ \widehat \L^{S,X}$ given by
\begin{align}\label{generatore}
\widehat \L^{S,X} f(t,s,x) &=  \frac{\partial f}{\partial t}+ \left [ b(t,x) - \rho \  { \mu(t,s,x) a(t,x) \over \sigma(t,s)}  \right ] \frac{\partial f}{\partial x} +
\frac{1}{2} a^2(t,x) \frac{\partial^2 f}{\partial
x^2}\\
&+ \rho a(t,x) \sigma(t,s) s  \frac{\partial^2 f}{\partial x\partial s}
  +\frac{1}{ 2} \sigma^2(t,s)\, s^2  \frac{\partial^2 f}{\partial s^2},
  \end{align}
for every function $f \in  \C^{1,2,2}_b([0,T]\times \R^+ \times \R )$. Moreover, the following semimartingale decomposition holds
\begin{equation}\label{semi}
f(t, S_t, X_t)=f(0,s_0,x_0)+\int_0^t \widehat \L^{S,X}f(u,S_u, X_u) \ud u + M^{f}_t, \quad t\in [0,T],
\end{equation}
where $M^{f}=\{M_t^{f},\ t \in [0,T]\}$ is the $(\bF, \widehat \P)$-martingale given by
 \begin{gather}
  \ud M^{f}_t =  \frac{\partial f}{\partial x} a(t,X_t) \left[\rho \ud \widehat W_t + \sqrt{1-\rho^2} \ud B_t\right]
  + \frac{\partial f}{\partial s} \sigma(t,S_t) S_t\ud \widehat {W}_t.
  \label{eq:M1f}
   \end{gather}
\end{lemma}

The proof is postponed to Appendix \ref{appendix:proofs}.

The idea for computing the $(\bF^S, \widetilde \bG)$-locally risk minimizing strategy is to derive $\theta^\F$ via \eqref{eq:thetaF} and apply equation \eqref{eq:theta*}. Therefore, we need to characterize the process $\widehat V^\G$ in \eqref{eq:hatVG}.

First, observe that the process $M$ in \eqref{eq:martingala_salto} is  $(\bG,\widehat \P)$-martingale null at time zero that can also be written as
\[M_t= H_t - \int_0^t (1-H_r) \gamma(r,X_r) \ud r,\]
where $H$ is the death indicator process given in \eqref{eq:death_process}, i.e. $H_t=\I_{\{\tau \le t\}}$. Then we get that,
\begin{align}
&\widehat V^\G_t=\widehatesp{ G(T,S^\tau_T) (1 - H_T) + \int_0^T U(r,S^\tau_r) \ud H_r | \G_t}\\
&\quad= \widehatesp{ G(T,S^\tau_T) (1 - H_T) + \int_0^T U(r,S^\tau_r)(1-H_r) \gamma(r,X^\tau_r) \ud r | \G_t}\\
&\quad= \int_0^t U(r,S^\tau_r)(1-H_r) \gamma(r,X^\tau_r) \ud r  + \widehatesp{ G(T,S^\tau_T) (1 - H_T) + \int_t^T U(r,S^\tau_r)(1-H_r) \gamma(r,X^\tau_r) \ud r | \G_t}.\label{eq:aspettazione}
\end{align}
In order to compute the last conditional expectation we use the Markovianity of the triplet $(S^{\tau},X^{\tau},H)$ under $\widehat \P$, which  is proved  in the lemma below. Denote by $\widehat \C_b^{1,2,2}$ the set of measurable and bounded functions $f:[0,T]\times \R^+ \times \R\times \{0,1\}\to \R $ which are continuous  and differentiable with respect to $t$, continuous and twice differentiable with respect to $(s,x)$ with bounded derivatives (of all necessary orders).

\begin{lemma}\label{markovbis}
Under conditions (ii) of Assumption \ref{ipot1} and \eqref{hp:generatore_cont}, the triplet $(S^{\tau},X^{\tau},H)$ is a $(\bG, \widehat \P)$-Markov process with generator $ \widehat \L^{S,X,H}$ given by
\begin{equation} \widehat \L^{S,X,H} f(t,s,x,z) = \widehat \L^{S,X} f(t,s,x,z) (1-z) + \{ f(t,s,x, z+1) - f(t,s,x,z)\} \gamma(t,x) (1-z)\end{equation}
for every function $f\in \widehat \C_b^{1,2,2}$, where $\widehat \L^{S,X} $ is given in \eqref{generatore}.

Moreover, the following $(\bG, \widehat \P)$-semimartingale decomposition holds
\begin{equation}\label{semibis}
f(t, S_t, X_t,H_t)=f(0,s_0,x_0,0)+\int_0^t \widehat \L^{S,X,H}f(u,S_u, X_u,H_u) \ud u + M^{f}_t, \quad t \in \llbracket0,T \wedge \tau\rrbracket,
\end{equation}
where $M^{f}=\{M_t^{f},\ t \in [0,T]\}$ is the $(\bG, \widehat \P)$-martingale given by
\begin{align}
  \ud M^{f}_t = & \frac{\partial f}{\partial x} (1-H_t) a(t,X_t) \left[\rho \ud \widehat W_t + \sqrt{1-\rho^2} \ud B_t\right]
  + \frac{\partial f}{\partial s} (1-H_t) \sigma(t,S_t) S_t\ud \widehat {W}_t\\
& +\{ f(t,S_t, X_t, H_{t^-}+1) - f(t,S_t,X_t,H_{t^-})\} \ud M_t.\end{align}
\end{lemma}

The proof is postponed to Appendix \ref{appendix:proofs}.

Then the following result provides a characterization of the locally risk-minimizing strategy for the insurance claim under full information.

\begin{proposition}[The full information case]\label{xi_bis}
Let  $g\in \C_b^{1,2,2}([0,T] \times \R^+ \times \R)$ be a solution of the problem
\begin{equation} \label{eq:problemabis}
\left\{
\begin{split}&\widehat  \L^{S,X} g(t,s,x)  - \gamma(t,x) g(t,s,x) + U(t,s) \gamma(t,x) = 0, \quad (t,s,x) \in [0,T) \times \R^+ \times \R\\
&g(T,s,x)  = G(T,s),
\end{split}\right.
\end{equation}
 then the $(\bF, \bG)$-locally risk minimizing strategy is given by
\begin{equation}\label{eq:thetaF_1}\theta^\F_t  = {\partial g \over \partial s}(t, S_t, X_t) + \frac{\rho  a(t, X_t)}{ S_t \sigma(t, S_t)} {\partial g \over \partial x}(t, S_t, X_t) \quad t \in \llbracket0,T \wedge \tau\rrbracket.\end{equation}

\end{proposition}
\begin{proof}
First, note that   if $g(t,x,s) \in \C_b^{1,2,2} ([0,T]\times \R \times \R^+)$ is  a solution of the problem
\eqref{eq:problemabis} then the function $\widehat g(t,x,s,z) \in  \widehat \C_b^{1,2,2} $, defined as  $\widehat g(t,x,s,0):= g(t,x,s)$ and  $\widehat g(t,x,s,1):=0$ solves  the backward Cauchy problem \begin{equation}\label{eq:problema}
\left\{
\begin{split}& \widehat  \L^{S,X,H} \widehat g(t,s,x,z)   + U(t,s) (1-z) \gamma(t,x) = 0, \quad (t,s,x,z) \in [0,T) \times \R^+ \times \R \times \{0,1\}\\
&\widehat g(T,s,x,z)  = (1-z)G(T,s).
\end{split}\right.
\end{equation}
By Lemma \ref{markovbis} and Feynman-Kac formula we have that
$$\widehat g(t,S_t,X_t,H_t) = \widehatesp{ G(T,S_T) (1 - H_T) + \int_t^T U(r,S_r)(1-H_r) \gamma(r,X_r) \ud r | \G_t}$$
and  the following $(\bG, \widehat \P)$-martingale decomposition of $\widehat V^\G$ holds,
\begin{align}
  \ud \widehat V^{\G}_t = & \frac{\partial \widehat  g}{\partial x} (1-H_t) a(t,X_t)\sqrt{1-\rho^2} \ud B_t +
   (1-H_t) \big [ \frac{\partial \widehat  g}{\partial s} \sigma(t,S_t) S_t +   \frac{\partial \widehat  g}{\partial x} a(t,X_t) \rho \big ] \ud \widehat {W}_t\\
& +\{  \widehat  g(t,S_t, X_t, H_{t^-}+1) -  \widehat  g(t,S_t,X_t,H_{t^-})\} \ud M_t.\end{align}

Then taking the predictable covariation of $\widehat V^\G$ with respect to $S^\tau$ we get
\begin{align}
\ud \langle \widehat V^\G, S^\tau \rangle^{\widehat \P}_t =& (1 - H_{t^-}) \left(S_t \sigma(t, S_t) {\partial \widehat  g \over \partial s}(t, S_t, X_t, H_{t^-}) + \rho  a(t, X_t) {\partial \widehat  g \over \partial x}(t, S_t, X_t, H_{t^-}) )\right) \ud t\\
=&(1 - H_{t^-}) \left(S_t \sigma(t, S_t) {\partial \widehat  g \over \partial s}(t, S_t, X_t, 0) + \rho  a(t, X_t)  {\partial \widehat  g \over \partial x}(t, S_t, X_t, 0) \right) \ud t
\end{align}

Since the predictable quadratic variation of $S^\tau$ satisfies
\begin{equation} \ud \langle S^\tau \rangle^{\widehat \P} _t = (1 - H_{t^-})  \ud \langle S \rangle _t = (1 - H_{t^-}) S_t \sigma(t,S_t) \ud t\end{equation}
we only need to apply relation \eqref{eq:thetaF} to get
\begin{equation}\label{eq:thetaF_2}\theta^\F_t  = {\partial \widehat g \over \partial s}(t, S_t, X_t, 0) + \frac{\rho  a(t, X_t)}{ S_t \sigma(t, S_t)} {\partial \widehat g \over \partial x}(t, S_t, X_t, 0) \quad t \in \llbracket0,T \wedge \tau\rrbracket,\end{equation}
which corresponds to \eqref{eq:thetaF_1}

\end{proof}

\begin{remark}
Existence and uniqueness of classical solutions to \eqref{eq:problemabis}  can be obtained under suitable assumptions by applying the results in \citet{heath2000martingales}.
\end{remark}

\begin{remark}
By  Feynmann-Kac formula the process $\{g(t,  S_t,X_t), \ t \in [0,T]\}$ has the following stochastic representation
\begin{align}\label{eq:FK_1} g(t,S_t,X_t)=\widehatesp{ e^{-\int_t^T\gamma(r, X_r)\ud r}G(T,S_T) + \int_t^T e^{-\int_t^r\gamma(u, X_u)\ud u} U(r,S_r)\gamma(r,X_r) \ud r | \F_t}. \end{align}
 \end{remark}

\subsection{A filtering approach to local risk-minimization under partial information}
In this section we wish to apply some results from filtering theory to compute the locally risk-minimizing hedging strategy under partial information.
Precisely, this requires to compute conditional expectations of processes that depend on the trajectories of $X$. To apply the classical methodology, we introduce as an additional state process the $\bF$-survival process of $\tau$ given by $\P(\tau>t|\F_t)=1-F_t= e^{- \int_0^t \gamma(u,X_u) \ud u}$, for each $t \in [0,T]$, and denote it by $Y_t$.
The dynamics of $Y=\{Y_t, \ t \in [0,T]\}$, is given by
\[
\ud Y_t= - \gamma(t, X_t) Y_t  \ud t, \quad Y_0=1.
\]

\begin{remark}
By performing the same computations of Lemma \ref{markov} for the triplet $(S,X,Y)$, it is easy to verify that the vector process $(S,X,Y)$ is an $(\bF, \widehat \P)$-Markov process with generator $ \widehat \L^{S,X,Y}$ given by
\begin{align}\label{eq:generatore2}
\widehat \L^{S,X,Y} f(t,s,x,y) = & \frac{\partial f}{\partial t}+ \left [ b(t,x) - \rho \  { \mu(t,s,x) a(t,x) \over \sigma(t,s)}  \right ] \frac{\partial f}{\partial x} -  y \gamma(t,x) \frac{\partial f}{\partial y} \\
& +\frac{1}{2} a^2(t,x) \frac{\partial^2 f}{\partial
x^2}+ \rho a(t,x) \sigma(t,s) s  \frac{\partial^2 f}{\partial x\partial s}
  +\frac{1}{ 2} \sigma^2(t,s)\, s^2  \frac{\partial^2 f}{\partial s^2},
\end{align}
for every function $f \in  \C^{1,2,2,1}_b([0,T]\times \R^+ \times \R \times \R^+)$.  Consequently, the process $\{f(t, S_t, X_t, Y_t), \ t \in [0,T]\}$ has the following semimartingale decomposition
\begin{equation}\label{eq:semimg_decomp}
f(t,S_t,X_t,Y_t)=f(0,s_0,x_0,1)+\int_0^t \widehat \L^{S,X,Y}f(u, S_u, X_u,Y_u) \ud u + M^{f}_t,\quad t \in [0,T],
\end{equation}
where $M^{f}=\{M_t^{f},\ t \in [0,T]\}$ is the $(\bF, \widehat \P)$-martingale given by
 \begin{gather}
  \ud M^{f}_t =  \frac{\partial f}{\partial x} a(t,X_t) \left[  \rho \ud \widehat W_t + \sqrt{1-\rho^2} \ud B_t\right]
  + \frac{\partial f}{\partial s} \sigma(t,S_t) S_t\ud \widehat {W}_t.
  \label{eq:M2f}
   \end{gather}
\end{remark}

For every measurable function $f(t,s,x,y)$ such that $\espp{|f(t,S_t,X_t,Y_t)|}< \infty$, for each $t\in [0,T]$, we define the {\em filter} $\pi(f)=\left\{\pi_t(f), \ t \in [0,T]\right\}$ with respect to $\widehat \P$, by setting
\begin{equation}\label{eq:filter}
\pi_t(f) := \widehatesp{ f(t,S_t, X_t,Y_t)  | \F^S_t}, \quad t \in [0, T].
\end{equation}
It is well known that $\pi$ is a probability measure-valued process with c\`{a}dl\`{a}g trajectories (see, e.g.~\citet{kurtz1988unique}), and provides the $\widehat \P$-conditional law of the stochastic factor $X$ given the filtration generated by the risky asset prices process.
The filter dynamics is given in Proposition \ref{filteq} below.
\begin{ass}\label{ass:uniq}
The functions $b$, $a$, $\gamma$, $\mu$, and $\sigma$ are jointly continuous and
satisfy the following growth and locally Lipschitz conditions:
\begin{itemize}
\item[(G)] for some nonnegative constant $C$, and for every $(t,s,x) \in [0,T]\times \R^+ \times \R$,
\begin{gather}
|b(t,x)|^2 + |a(t,x)|^2  + |\gamma(t,x)|^2  \leq C( 1 + |x|^2), \\
 |\mu(t,s,x)| ^2 \leq C( 1 + |s|^2 + |x|^2) \mbox{ and } |\sigma(t,s)|^2 \leq C(1 + |s|^2);
\end{gather}
\item[(LL)]
for all $r>0$ there exists a constant $L$ such that for every $t \in [0,T]$, $ s, s^\prime,x,x^\prime \in B_r(0) := \{ z \in \R: \quad |z| \leq r\}$,
\begin{gather}
|b(t,x) - b(t,x^\prime)| + |a(t,x) - a(t,x^\prime)| +  |\gamma(t,x) - \gamma(t,x^\prime)| \leq L|x- x^\prime|,\\
|\mu(t,s,x) - \mu(t,s^\prime,x^\prime)|\leq L (  |s- s^\prime|+ |x- x^\prime| ) \mbox{ and } |\sigma(t,s) - \sigma(t,s^\prime)|\leq L|s- s^\prime|.
\end{gather}
\end{itemize}
\end{ass}

\begin{proposition}\label{filteq}
Under Assumption \ref{ipot1}, Assumption \ref{ass:uniq} and condition \eqref{hp:generatore_cont},  for every function $f \in \C_b^{1,2,2,1}([0,T] \times \R^+ \times \R \times \R^+)$  and $t \in [0,T]$,
the filter $\pi$ is the unique strong solution of the following equation
\begin{equation} \label{eq:ks}
\pi_t (f) = f(0,s_0,x_0,1) + \int_0^t \pi_u(\widehat \L^{S,X,Y} f) \ud u + \int_0^t \left [ \rho \pi_{u}\left(a \  \frac{\partial f }{\partial x}\right) + S_u \sigma(t,S_u) \pi_u \left(\frac{\partial f }{\partial
 s}\right) \right ] \ud \widehat {W}_u.  \end{equation}

\end{proposition}

The proof is postponed to Appendix \ref{appendix:proofs}.

Now, we can characterize the optimal hedging strategy for the given endowment insurance contract $(\xi, Z, \tau)$  under partial information as follows.
\begin{theorem}\label{theta*}
Assume that $S^\tau$ and $N_{T\wedge \tau}$ are $\widehat \P$-square integrable. Let $g$ be a solution to problem \eqref{eq:problemabis}.
The first component $\theta^*$ of the $(\bF^S, \widetilde \bG)$-locally risk-minimizing strategy for the payment stream $N$ associated to the unit-linked life insurance contract  $(\xi, Z, \tau)$
 is given by
\begin{equation}\label{eq:theta_bis}
\theta^*_t=  \frac{\pi_{t} \left( id_y {\partial  g \over \partial s}\right) +
{\rho \over \sigma(t,S_t) S_t}  \pi_{t} \left(a\ id_y {\partial  g \over \partial x} \right)
}{\pi_{t}(id_y)},
\end{equation}
for every $t \in \llbracket 0, T\wedge\tau\rrbracket$, where $id_y(t,s,x,y) := y$.
\end{theorem}

\begin{proof}
By equation \eqref{eq:theta*} in Theorem \ref{thm:strategy} and \eqref{eq:thetaF_1}  we get
\begin{align}\theta^*_t& =  \frac{{}^{\widehat p, \bF^S} \left(\theta^\F_t  e^{- \int_0^t \gamma(u,X_u) \ud u}\right)}{{}^{\widehat p, \bF^S} ( e^{- \int_0^t \gamma_u \ud u} ) }\\
& = \frac{{}^{\widehat p, \bF^S} \left ( e^{- \int_0^t \gamma(u,X_u) \ud u}  {\partial  g \over \partial s}(t, S_t, X_t) \right)}{{}^{\widehat p, \bF^S} ( e^{- \int_0^t \gamma_u \ud u} )}  + \frac{{}^{\widehat p, \bF^S} \left (e^{- \int_0^t \gamma(u,X_u) \ud u} \frac{\rho  a(t, X_t)}{S_t \sigma(t, S_t) } {\partial  g\over \partial x}(t, S_t, X_t)   \right )}{{}^{\widehat p, \bF^S} ( e^{- \int_0^t \gamma_u \ud u} ) },
\end{align}
for every $t \in \llbracket 0, T\wedge\tau\rrbracket$. Finally,  \eqref{eq:theta_bis} follows by the definition of the filter.
\end{proof}

\subsection{An example with uncorrelated Brownian motions}

Throughout this section we choose $\rho=0$, which corresponds to the case where $W$ and $B$ are $\P$-independent, and therefore $\widehat W$ and $B$ are $\widehat \P$-independent. In this case, a simpler expression for the first component of the optimal hedging strategy $\theta^*$ under partial information is provided.

On the probability space $(\Omega, \F, \widehat \P)$ the dynamics of the vector process $(S,X, Y)$ is given by
\begin{equation}\label{eq:uncorrelated}
\left\{
\begin{split}
&\ud S_t=S_t \sigma(t, S_t) \ud \widehat W_t, \quad S_0=s_0>0,\\
&\ud X_t=b(t, X_t)\ud t + a(t,X_t)\ud B_t, \quad X_0=x_0\in \R,\\
&\ud Y_t=-Y_t\gamma(t,X_t)\ud t, \qquad Y_0=1.
\end{split}\right.
\end{equation}

Moreover, we choose a recovery function of the form $U(t,s)=\delta \ s$, for every $(t,s)\in [0,T]\times \R^+$, where $\delta$ is a given positive constant. Then, the payment stream $N$ is given by $\ds N_t=\delta\int_0^t S_u \ud H_u$ if $t \in [0,T)$ and $\ds N_T = G(T,S_T)\I_{\{ \tau > T \}}$.

In the sequel we wish to characterize the optimal hedging strategy under full information, given in  \eqref{eq:thetaF_1}, and under partial information via \eqref{eq:theta*}, in this simpler example. This requires to compute $g$ in equation \eqref{eq:FK_1}.

The independence between $X$ and $S$ under $\widehat \P$ (that also holds when conditioning to $\F_t$, for each $t$), implies that
\begin{align}
&\widehatesp{G(T, S_T)e^{-\int_t^T\gamma(u, X_u)\ud u}\Big{|}\F_t}=\widehatesp{G(T, S_T)|\F_t}\widehatesp{e^{-\int_t^T\gamma(u, X_u)\ud u}\Big{|}\F_t} \\ &\qquad=\widetilde{g}(t,S_t)\frac{\widehatesp{Y_T|\F_t}}{Y_t}=\widetilde{g}(t,S_t)\widehatesp{Y_T|\F_t}e^{\int_0^t\gamma(r, X_r) \ud r},
\end{align}
for every $t \in [0,T]$, where by the Feynman-Kac theorem the function $\widetilde g$  can be characterized as the solution of the problem
\begin{equation}
\left\{
\begin{aligned}& \frac{\partial \widetilde g}{\partial t}(t,s)+\frac{\partial^2 \widetilde g}{\partial s^2}(t,s)\sigma^2(t,s) s^2=0, \quad (t,s) \in [0,T) \times \R^+,\\
&\widetilde g(T,s)=G(T,s)  , \quad s\in \R^+.
\end{aligned}
\right.
\end{equation}
Then for the remaining part of the conditional expectation in \eqref{eq:FK_1}, using the $\widehat \P$-independence between $(X, Y)$ and $S$ and the fact that $S$ is an $(\bF,  \widehat P)$-martingale, we have
\begin{align}
&\delta \widehatesp{\int_t^T e^{-\int_t^r\gamma(u, X_u)\ud u} S_r\gamma(r,X_r) \ud r | \F_t}=\delta\widehatesp{\int_t^T \frac{Y_r}{Y_t} S_r\gamma(r,X_r) \ud r | \F_t}\\
&\qquad=-\frac{\delta}{Y_t}\widehatesp{\int_t^T  S_r \ud Y_r | \F_t}=-\frac{\delta}{Y_t}\widehatesp{\int_t^T \ud( S_r  Y_r) | \F_t}+\frac{\delta}{Y_t}\widehatesp{\int_t^T  Y_r \ud S_r| \F_t}\\
&\qquad=-\frac{\delta}{Y_t}\widehatesp{S_TY_T-S_tY_t| \F_t}=\frac{\delta S_t}{Y_t}\left(Y_t-\widehatesp{Y_T|\F_t}\right).
\end{align}

This implies that
\[
g(t, X_t, S_t)=\widetilde{g}(t,S_t) \widehatesp{Y_T|\F_t}e^{\int_0^t\gamma(r, X_r) \ud r}+\frac{\delta S_t}{Y_t}\left(Y_t-\widehatesp{Y_T|\F_t}\right)
\]
\[
= ( \widetilde{g}(t,S_t)  - \delta S_t) e^{\int_0^t\gamma(r, X_r) \ud r}  \widehatesp{Y_T|\F_t} +\delta S_t.
\]

\begin{remark}
Note that for every $t \in [0,T]$,
\[
\widehatesp{Y_T|\F_t}=e^{-\int_0^t\gamma(u, X_u)\ud u}\widehatesp{e^{-\int_t^T\gamma(u, X_u)\ud u}\Big{|}\F_t}=e^{-\int_0^t\gamma(u, X_u)\ud u}\widehatesp{e^{-\int_t^T\gamma(u, X_u)\ud u}\Big{|}\F^B_t},
\]
where the last equality follows by the independence of $X$ and $W$ under $\widehat \P$. By the Feynman-Kac theorem, if there exists a function $\Phi \in \C_b^{1,2}([0,T] \times \R)$ which solves the problem
\begin{equation}
\left\{
\begin{aligned}& \frac{\partial \Phi}{\partial t}(t,x)+\frac{\partial \Phi}{\partial x}(t,x)b(t,x)+\frac{1}{2}\frac{\partial^2 \Phi}{\partial x^2}(t,x)a^2(t,x) -\frac{\partial \Phi}{\partial y}(t,x)\gamma(t,x)=0, \quad (t,x) \in [0,T) \times \R, \\
&\Phi(T,x)=1  , \quad x\in \R,
\end{aligned}
\right.
\end{equation}
then, $\ds \Phi(t,X_t)=\widehatesp{e^{-\int_t^T\gamma(u, X_u)\ud u}\Big{|}\F^B_t}$ and the process $\ds \left\{e^{-\int_0^t\gamma(u, X_u)\ud u} \Phi(t, X_t), \  t \in [0,T]\right\}$ is an $(\bF^B, \widehat \P)$-martingale.

\end{remark}

Hence $g(t,S_t,X_t) = \widetilde{g}(t,S_t) \Phi(t,X_t) +  \delta S_t(1- \Phi(t,X_t) )$ and by using  \eqref{eq:thetaF_1} the optimal hedging strategy under full information is given by
\begin{equation}\label{eq:strategia_esempiofull}
\theta^{\F}_t=\left(\frac{\partial \widetilde g}{\partial s}(t, S_t) -\delta\right)  \Phi(t,X_t)+\delta, \quad t \in \llbracket0,T\wedge \tau\rrbracket.
\end{equation}

Finally, by  \eqref{eq:theta*} we get that the $(\bF^S, \widetilde \bG)$-locally risk-minimizing strategy can be written as
\begin{equation}\label{eq:strategia_esempio}
\theta^*_t=\frac{\left(\frac{\partial \widetilde g}{\partial s}(t, S_t) - \delta\right) \pi_{t}\left(id_y \ \Phi \right)}{\pi_{t}\left(id_y\right)}+\delta ,\quad t \in \llbracket0,T\wedge \tau\rrbracket.
\end{equation}

Note that, by the $\widehat \P$-independence of $(X, Y)$ and $S$, and the fact that the change of probability measure from $\P$ to $\widehat \P$ does not affect the law of $X$,  we have that the computation of the filter reduces to ordinary expectations with respect to $\P$
\begin{gather}
\pi_t( \Phi \ id_y) = \widehatesp{\Phi(t,X_t) e^{-\int_0^t \gamma(u,X_u) \ud u}\Big{|} \F^S_t}=\widehatesp{\Phi(t,X_t) e^{-\int_0^t \gamma(u,X_u) \ud u}}=\Phi(0, x_0)=\esp{Y_T},\\
\pi_t(id_y) = \widehatesp{e^{-\int_0^t \gamma(u,X_u) \ud u}\Big{|} \F^S_t}=\widehatesp{e^{-\int_0^t \gamma(u,X_u) \ud u}}=\widehatesp{Y_t}=\esp{Y_t},
\end{gather}
for every $t \in [0,T]$. Then we can write \eqref{eq:strategia_esempio} as
\begin{equation}\label{eq:strategia_esempio_2}
\theta^*_t=\frac{\left(\frac{\partial \widetilde g}{\partial s}(t, S_t) - \delta\right) \esp{Y_T} + \delta \esp{Y_t} }{\esp{Y_t}}, \quad t \in \llbracket0,T\wedge \tau\rrbracket
\end{equation}

where $\esp{Y_t} = \esp{ e^{-\int_0^t \gamma(u,X_u) \ud u}}$, $t \in [0,T]$.

 \begin{center}
{\bf Acknowledgements}
\end{center}
The authors wish to acknowledge Professor Monique Jeanblanc for valuable comments and suggestions. The authors are also grateful to the Gruppo Nazionale per l'Analisi Matematica, la Probabilit\`{a} e le loro
Applicazioni (GNAMPA) of the Istituto Nazionale di Alta Matematica (INdAM) for the financial support through project number 2016/000371.

\appendix

\section{The hazard process and the martingale hazard process of $\tau$ under partial information}\label{appendix:hazard_rate}

We define the conditional distribution of $\tau$ with respect to $\F^S_t$, for every $t \in [0,T]$, as
\begin{equation}\label{eq:Ftilde}
  F^S_t = \P(\tau \leq t |\F^S_t),\quad t \in [0,T].
\end{equation}
By the tower rule it is easy to check that $F^S_t=\esp{F_t|\F^S_t}$, for each $t \in [0,T]$. Hence, the assumption $F_t < 1$, for every  $ t \in [0,T]$, also implies that $F^S_t < 1$ for every $t \in [0,T]$.

We now introduce the $\bF^S$-hazard process of $\tau$ under $\P$,  $\Gamma^S=\{\Gamma^S_t, \ t \in [0,T]\}$, by setting
\begin{equation}\label{eq:hazard process2}
\Gamma^S_t=-\ln(1-F^S_t), \quad t \in [0,T].
\end{equation}

\begin{remark}
Notice that the relationship between the $\bF$-hazard process $\Gamma$, see \eqref{eq:hazard process}, and the $\bF^S$-hazard process $\Gamma^S$, see \eqref{eq:hazard process2}, is given by
\begin{equation}\label{eq:relazione}
e^{-\Gamma^S_t}=\esp{e^{-\Gamma_t}|\F^S_t}, \quad t \in [0,T].
\end{equation}
 \end{remark}
If  $\Gamma^S$ is continuous and increasing, by \citet[Proposition 5.1.3]{bielecki2002}
the process
$\{H_t-\Gamma^S_{t\wedge \tau},\ t \in [0,T]\}$ is a $(\widetilde \bG, \P)$-martingale.
However, without these assumptions, we will prove in Proposition \ref{prop:mg_haz_proc}
the existence of  an $(\bF^S, \widetilde \bG)$-martingale hazard process.

For the sake of clarity, we recall the definition of martingale hazard process in our setting.

\begin{definition}
An $\bF^S$-predictable, increasing process $\Lambda=\{\Lambda_t,\ t \in [0,T]\}$ is called an $(\bF^S, \widetilde \bG)$-martingale hazard process of the random time $\tau$ if and only if the process $\{H_t - \Lambda_{t\wedge \tau},\ t \in [0,T]\}$ follows a $(\widetilde \bG,\P)$-martingale.
\end{definition}

In general,  the $(\bF^S, \widetilde \bG)$-martingale hazard process  does not coincide with the $\bF^S$-hazard process $\Gamma^S$. This property is fulfilled  if the martingale invariance property holds, that is, any
$(\bF^S, \P)$-martingale turns out to be a $(\widetilde \bG, \P)$-martingale. In such a case, the $(\bF^S, \widetilde \bG)$-martingale hazard process uniquely specifies the $\bF^S$-survival probabilities of $\tau$. Nevertheless, we do not make this assumption in the paper.

In order to derive the $(\bF^S, \widetilde \bG)$-martingale hazard process of $\tau$ we need some preliminary results.

Recall that
given any subfiltration $\bH=\{\H_t,\ t \in [0,T]\}$ of $\bG$,

${}^{o,\bH} {Y}$ (respectively ${}^{p,\bH} Y$) denotes the optional (respectively predictable) projection of a given  $\P$-integrable, $\bG$-adapted process $Y$ with respect to $\bH$ and $\P$.

\begin{lemma}\label{relation}
Given a $\P$-integrable, $\bG$-adapted process $Y$, we have
\begin{align}
\I_{\{\tau > t\}}  {}^{o, \widetilde \bG} Y_t & = \I_{\{\tau > t\}} \frac{ {}^{o, \bF^S}\left(Y_t  \I_{\{\tau > t\}}\right)} {{}^{o, \bF^S} { \I_{\{\tau > t\}}}}, \label{op1}\\
\I_{\{\tau \geq t\}}  {}^{p,\widetilde \bG} Y_t & = \I_{\{\tau \geq t\}} \frac{{}^{p, \bF^S} \left(Y_t  \I_{\{\tau \geq t\}}\right)}{{}^{p, \bF^S}\I_{\{\tau \geq t\}}}, \label{pr1}
\end{align}
for each $t \in [0,T]$. Moreover, if $Y$ is $\bF$-predictable then
\begin{equation}\label{pr2}
\I_{\{\tau \geq t\}}  {}^{p,\widetilde \bG} Y_t = \I_{\{\tau \geq t\}}\frac{{}^{p, \bF^S} \left(Y_t  e^{- \int_0^t \gamma_u \ud u}\right)} {{}^{p, \bF^S} (e^{- \int_0^t \gamma_u \ud u})}, \quad t \in [0,T].
 \end{equation}

\end{lemma}

\begin{proof}
How to get formula \eqref{op1} is shown in \citet[Lemma 5.1.2]{bielecki2002}.

To prove \eqref{pr1}, first observe that  since $F_t<1$ for all $t\in [0,T]$,  there exists an $\bF^S$-predictable process $\widetilde Y=\{\widetilde Y_t, \ t \in [0,T]\}$ such that  $\widetilde Y_t  \  \I_{\{\tau\geq t\}} = {}^{p,\widetilde \bG} Y_t \  \I_{\{\tau\geq t\}}$, $\P$-a.s. for each $t \in [0,T]$.
By the predictable projection properties, for any $\bF^S$-predictable process $\varphi=\{\varphi_t, \ t \in [0,T]\}$ and for each $t \in [0,T]$, we get
\begin{align}
\esp{\int_0^t \varphi_s\widetilde Y_s \ {}^{p, \bF^S} \!  \I_{\{\tau \geq s\}}  \ud s} & = \esp{\int_0^t \varphi_s \widetilde Y_s   \I_{\{\tau \geq s\}} \ud s} = \esp{\int_0^t \varphi_s \I_{\{\tau \geq s\} } \ {}^{p,\widetilde \bG} Y_s \ud s} \\
&=  \esp{\int_0^t \varphi_s \I_{\{\tau \geq s\} } Y_s \ud s} = \esp{\int_0^t \varphi_s \ {}^{p, \bF^S} \!\left(\I_{\{\tau \geq s\}} Y_s\right)  \ud s},
\end{align}
since the process $\{ \varphi_t  \I_{\{\tau \geq t\} },  \ t \in [0,T]\}$ is $\widetilde \bG$-predictable.

Now consider the case where  $Y$ is $\bF$-predictable. 
Since $\{{}^{o, \bF} \I_{\{\tau > t\}}=e^{- \int_0^t \gamma_u \ud u},\ t \in [0,T]\}$ is a continuous process, we get
$${}^{o, \bF} \I_{\{\tau > t\}} = {}^{o, \bF} \I_{\{\tau \geq t\}} =  {}^{p, \bF} \I_{\{\tau \geq t\}},\quad t \in [0,T].
$$

Finally, equation 
 \eqref{pr2} is consequence of the following chains of equalities
$${}^{p, \bF^S} \I_{\{\tau \geq t\}} = {}^{p, \bF^S} \left( {}^{p, \bF} \I_{\{\tau \geq t\}}\right) = {}^{p, \bF^S}\left( e^{- \int_0^t \gamma_u \ud u}\right),
$$
and
$$ {}^{p, \bF^S} \left(Y_t  \I_{\{\tau \geq t\}}\right)  =
{}^{p, \bF^S} \left(Y_t   \ {}^{p, \bF} \I_{\{\tau \geq t\}}\right) = {}^{p, \bF^S} \left(Y_t  e^{- \int_0^t \gamma_u \ud u}\right),
 $$
 for every $t \in [0,T]$.
\end{proof}

\begin{remark}
Note that the $\bF^S$-hazard process $\Gamma^S=\{\Gamma^S_t, \ t \in [0,T]\}$, can be written as
\begin{equation}\label{eq:gammatilde}
\Gamma^S_t=-\ln\left({}^{o, \bF^S} \left(e^{- \int_0^t \gamma_u \ud u}\right)\right), \quad t \in [0,T].
\end{equation}
\end{remark}

\begin{remark}\label{projection}
Given a $(\bG, \P)$-martingale $m =\{m_t,  \ t \in [0,T]\}$ and a $\bG$-progressively measurable process $\psi = \{ \psi_t,  \ t \in [0,T]\}$ such that $\esp{\int_0^T |\psi_u| \ud u} < \infty$, the processes $\ds \left\{{}^{o,\widetilde \bG} m_t = \esp{m_t | \widetilde \G_t}, \ t \in [0,T]\right\}$  and $\ds \left\{{\phantom{\bigg{|}}}^{o,\widetilde \bG} \left( \int_0^t \psi_u \ud u\right) -  \int_0^t {}^{o,\widetilde \bG} \psi_u \ud u = \esp{ \int_0^t \psi_u \ud u \Big{|} \widetilde \G_t}  - \int_0^t  \esp{   \psi_u  \Big{|} \widetilde \G_u} \ud u, \ t \in [0,T]\right\}$ are  $(\widetilde \bG, \P)$-martingales, see for instance \citet[Remark 2.1]{cecicolaneri2012}.
\end{remark}

Finally, we give the $(\bF^S, \widetilde \bG)$-martingale hazard process of $\tau$.

\begin{proposition} \label{prop:mg_haz_proc}
The death time $\tau$ admits an  $(\bF^S, \widetilde \bG)$-martingale hazard process $\Lambda =\{\Lambda_t, \ t \in [0,T]\}$, where $\Lambda_t := \int_0^{t } \gamma^S_u \ud u $, with $\gamma^S=\{\gamma^S_t, \ t \in [0,T]\}$ being a nonnegative, $\bF^S$-predictable process.
Moreover, for every $t \in [0,T]$,
\begin{equation}\label{gtilde-intensity2}
\gamma^S_t  \I_{\{\tau\geq t\}} = {}^{p,\widetilde \bG} \gamma_t  \I_{\{\tau\geq t\}} \quad \P-a.s.
\end{equation}
and
\begin{equation}\label{esplicita}
  \gamma^S_t =
  {{}^{p, \bF^S} \left(\gamma_t  e^{- \int_0^t \gamma_u \ud u}
  \right)  \over {}^{p, \bF^S} \left(e^{-\int_0^t \gamma_u \ud u}\right) }, \quad t \in \llbracket 0, T \wedge \tau\rrbracket.
\end{equation}
\end{proposition}

\begin{proof}
By applying  Remark \ref{projection} to the   $(\bG, \P)$-martingale $M$, see \eqref{eq:martingala_salto},
we have that
$$
\left\{H_t -\int_0^{t} {}^{o,\widetilde \bG} \lambda_u \ud u,\ t \in [0,T]\right\}
$$
is a $(\widetilde \bG, \P)$-martingale,
which implies, taking Lemma \ref{op-pre} into account, that also
$$
\left\{H_t -\int_0^{t} {}^{p,\widetilde \bG} \lambda_u \ud u =  H_t - \int_0^{t \wedge \tau} {}^{p,\widetilde \bG} \gamma_u \ud u,\ t \in [0,T]\right\}
$$
is a $(\widetilde \bG, \P)$-martingale.

Since $F_t<1$ for all $t\in  [0,T]$, for any
$\widetilde \bG$-predictable process  $h=\{h_t,\ t \in [0,T]\}$  there exists
an $\bF^S$-predictable process $\widetilde h=\{\widetilde h_t,\ t \in [0,T]\}$
 such that $\widetilde h_t \I_{\{\tau\geq t\}} = h_t  \I_{\{\tau\geq t\}}$, $ \P$-a.s. for each $t \in [0,T]$.
This implies the existence of an
  $\bF^S$-predictable process $\gamma^S$ such that
   \eqref{gtilde-intensity2} is satisfied.

Hence, the process $\{\Lambda_t = \int_0^t  \gamma^S_u  \ud u,\ t \in [0,T]\}$ is  an $(\bF^S, \widetilde \bG)$-martingale hazard process of $\tau$ since $H_t - \Lambda_{t \wedge \tau} =  H_t - \int_0^{t \wedge \tau} \gamma^S_u  \ud u$, for each $t \in [0,T]$,
is a $(\widetilde \bG, \P)$-martingale.
To complete the proof is sufficient to apply the relationship \eqref{pr2} in Lemma \ref{relation}.
\end{proof}

Note that Proposition \ref{prop:mg_haz_proc} ensures that $\tau$ turns out to be a totally inaccessible $\widetilde \bG$-stopping time thanks to \citet[Chapter 6.78]{dellacherie_meyer1982}.

\section{Technical results} \label{appendix:projections}
\subsection{On optional and predictable projections under partial information}

\begin{lemma}\label{op-pre}
Given a  $\bG$-progressively measurable process $\psi= \{ \psi_t,  \ t \in [0,T]\}$ such that $\esp{\int_0^{T} |\psi_u| \ud u} < \infty$, then
$$
\int_0^{t } {}^{o,\widetilde \bG} \psi_u \ud u = \int_0^{t }  {}^{p,\widetilde \bG} \psi_u \ud u  \quad \P-a.s. \quad  t \in [0, T].
$$
\end{lemma}

\begin{proof}
First, we prove that the process $U=\{U_t,\ t \in [0,T]\}$ given by $U_t := \int_0^{t } ({}^{o,\widetilde \bG} \psi_u -  {}^{p,\widetilde \bG} \psi_u)  \ud u$, $t \in [0, T ]$,
is a $(\widetilde \bG, \P)$-martingale.

By the properties of predictable and optional projections, for any $\widetilde \bG$-predictable process $\varphi=\{\varphi_t, \ t \in [0,T]\}$ we get
\begin{equation}
\esp{\int_0^T \varphi_u {}^{p, \widetilde \bG} \psi_u  \ud u} = \esp{\int_0^T \varphi_u \psi_u  \ud u} = \esp{\int_0^T \varphi_u {}^{o, \widetilde \bG} \psi_u  \ud u}.
 \end{equation}
By choosing $\varphi_u = \I_A \I_{(s,t]}(u)$, $s <t$, $A \in \widetilde \G_s$, we obtain that
\begin{equation}
\esp{ \I_A \int_s^t ( {}^{p, \widetilde \bG} \psi_u  -  {}^{o, \widetilde \bG} \psi_u  ) \ud u} =0.
 \end{equation}

Finally, since $U$ is a process of finite variation by \citet[Chapter IV, Proposition 1.2]{revuz2013continuous}, $U$ is necessarily constant and equal to $U_0=0$, which concludes the proof.
\end{proof}

For reader's convenience, we provide a version of the Kallianpur-Striebel formula holding for predictable projections.
\begin{lemma} \label{lem:KS-pred}
If $G=\{G_t,\ t \in [0,T]\}$ is an $\bF$-adapted process, such that $\esp{G_t L_t} < \infty$, for any $t \in [0,T]$, then
\begin{equation} \label{eq:KS_pred}
{}^{\widehat p, \bF^S} G_t = \frac{{}^{p, \bF^S}\left(G_t L_t\right)}{{}^{p, \bF^S} L_t}, \quad t \in [0,T],
\end{equation}
where $L$ is the density process given in \eqref{eq:densitaL}.
\end{lemma}

\begin{proof}
To prove the result, we need to check that for every $\bF^S$-predictable process $\varphi$, the following equality holds
\begin{equation}
\widehatesp{\int_0^t \varphi_s  \ {}^{\widehat p, \bF^S} G_s\ {}^{p, \bF^S} L_s \ud s}=\widehatesp{\int_0^t \varphi_s \  {}^{p, \bF^S}\left(G_s L_s\right)\ud s},
\end{equation}
for every $t \in [0,T]$. By applying Fubini's theorem twice, and the property of the predictable projection, for every $\bF^S$-predictable process $\varphi$ and for every $t \in [0,T]$, we get
\begin{align*}
\widehatesp{\int_0^t \varphi_s  \ {}^{\widehat p, \bF^S} G_s \ {}^{p, \bF^S} L_s \ud s} & = \widehatesp{\int_0^t \varphi_s  G_s\ {}^{p, \bF^S} L_s \ud s}=\int_0^t \widehatesp{\varphi_s  G_s\ {}^{p, \bF^S} L_s}\ud s\\
&=\int_0^t \esp{\varphi_s  G_s L_s \ {}^{p, \bF^S} L_s}\ud s=\int_0^t \esp{\varphi_s \  {}^{p, \bF^S}\left(G_s L_s\right)\ {}^{p, \bF^S} L_s}\ud s\\
&= \int_0^t\esp{\varphi_s  \ {}^{p, \bF^S}\left(G_s L_s\right) L_s}\ud s=\widehatesp{\int_0^t\varphi_s \ {}^{p, \bF^S}\left(G_s L_s\right)\ud s},
\end{align*}
which concludes the proof.
\end{proof}

If the process $G$ is $\bG$-adapted but not necessarily $\bF$-adapted, then a similar result is showed in the following lemma.

\begin{lemma}\label{lem:KS-pred2}
If $G=\{G_t,\ t \in [0,T]\}$ is a $\bG$-adapted process, such that $\esp{G_t L_t} < \infty$, for any $t \in [0,T]$, then
\begin{equation} \label{eq:KS_pred2}
\I_{\{\tau\ge t\}}{}^{\widehat p, \widetilde \bG} G_t = \I_{\{\tau\ge t\}} \frac{{}^{p, \widetilde \bG}\left(G_t L_t\right)}{{}^{p, \widetilde \bG} L_t}, \quad t \in [0,T].
\end{equation}
\end{lemma}

\begin{proof}
Similarly to the proof of Lemma \ref{lem:KS-pred}, for every $\bG$-adapted process $G$ and every $\widetilde \bG$-predictable process $\varphi$ we have
\begin{align*}
&\widehatesp{\int_0^t \I_{\{\tau\ge s\}} \varphi_s  \ {}^{\widehat p, \widetilde \bG} G_s \ {}^{p, \widetilde \bG} L_s \ud s}  = \widehatesp{\int_0^t \I_{\{\tau\ge s\}} \varphi_s  G_s\ {}^{p, \widetilde \bG} L_s \ud s}=\int_0^t \widehatesp{\I_{\{\tau\ge s\}} \varphi_s  G_s\ {}^{p, \widetilde \bG} L_s}\ud s\\
&\hspace{1cm} =\int_0^t \esp{L^\tau_s \I_{\{\tau\ge s\}} \varphi_s  G_s \ {}^{p, \widetilde \bG} L_s}\ud s=\int_0^t \esp{\I_{\{\tau\ge s\}} \varphi_s \  {}^{p, \widetilde \bG}\left(G_s L_s\right)\ {}^{p, \widetilde \bG} L_s}\ud s\\
&\hspace{1cm}= \int_0^t\esp{\I_{\{\tau\ge s\}}\varphi_s  \ {}^{p, \widetilde \bG}\left(G_s L_s\right) L_s}\ud s=\widehatesp{\int_0^t \I_{\{\tau\ge s\}} \varphi_s \ {}^{p, \widetilde \bG}\left(G_s L_s\right)\ud s},
\end{align*}
for every $t \in [0,T]$. Note that, in the second line, we use the fact that $L^\tau_t=L_t$ for every $t \in \llbracket 0,T\wedge\tau\rrbracket$, where $L^\tau$ is the density process given in \eqref{eq:Ltau}.
\end{proof}

\begin{corollary}\label{cor:pred_proj}
Let $\theta=\{\theta_t,\ t \in [0,T]\}$ be an $\bF$-predictable process. Then,
\begin{equation}
\I_{\{\tau\geq t\}} \ {}^{\widehat p, \widetilde \bG}\theta_t = \I_{\{\tau\geq t\}} \frac{{}^{\widehat p, \bF^S}(\theta_t e^{-\int_0^t\gamma_u \ud u})}{{}^{\widehat p, \bF^S}(e^{-\int_0^t\gamma_u \ud u})}, \quad t \in [ 0,T].
\end{equation}
\end{corollary}

\begin{proof}
By Lemma \ref{lem:KS-pred2} we get
\begin{align}
\I_{\{\tau\ge t\}}{}^{\widehat p, \widetilde \bG} \theta_t &= \I_{\{\tau\ge t\}} \frac{{}^{p, \widetilde \bG}\left(\theta_t L_t\right)}{{}^{p, \widetilde \bG} L_t}\\
&\label{eq:second_line}= \I_{\{\tau\ge t\}}\frac{{}^{p, \bF^S}\left(\theta_t L_te^{-\int_0^t\gamma_u\ud u}\right)}{{}^{p, \bF^S} \left( e^{-\int_0^t\gamma_u\ud u}\right)} \cdot \frac{{}^{p, \bF^S} \left( e^{-\int_0^t\gamma_u\ud u}\right)}{{}^{p, \bF^S} \left(e^{-\int_0^t\gamma_u\ud u}L_t\right)}\\
&\label{eq:third_line} =\I_{\{\tau\geq t\}} \frac{{}^{\widehat p, \bF^S}(\theta_t e^{-\int_0^t\gamma_u \ud u})}{{}^{\widehat p, \bF^S}(e^{-\int_0^t\gamma_u \ud u})}
\end{align}
where in line \eqref{eq:second_line} we use Lemma \ref{relation} and in line \eqref{eq:third_line} we apply Lemma \ref{lem:KS-pred}.
\end{proof}

\subsection{Some proofs}\label{appendix:proofs}

\begin{proof}[Proof of Lemma \ref{lemma:stopped_inovation}]
First, note that $I^\tau$ is a square integrable process with continuous trajectories, and since the following equality is fulfilled
\begin{align}
I^\tau_t  = \int_0^t {1 \over \sigma(u,S^\tau_u) S^\tau_u} \ \ud S^\tau_u  - \int_0^t { {}^{p,\widetilde \bG} \mu_u\over \sigma(u,S^\tau_u)} \  \  \ud u, \quad t \in [0,T],
\end{align}
it turns out to be $\widetilde \bG$-adapted. We now prove that $I^\tau$ is a $(\widetilde \bG, \P)$-martingale. As a consequence of Lemma \ref{op-pre} in Appendix \ref{appendix:projections}, we can work with the $(\widetilde \bG, \P)$-optional projection of $\mu$, that is ${}^{o,\widetilde \bG} \mu$,  instead of the $(\widetilde \bG, \P)$-predictable projection $ {}^{p,\widetilde \bG} \mu$.
Hence, for every $0\leq s \leq t \leq T$, we have
$$
\esp{I^\tau_t   - I^\tau_s   \Big{|} \widetilde \G_s}  =  \esp{\int_{s\wedge \tau}^{t\wedge \tau}  {\mu(u, S^\tau_u, X_u^\tau) - {}^{o,\widetilde \bG} \mu_u \over \sigma(u,S^\tau_u) } \  \ud u  \bigg{|} \widetilde \G_s} + \esp{W^\tau_t  - W^\tau_s \Big{|} \widetilde \G_s}.
$$
By the properties of the conditional expectation
 we obtain that
\begin{equation}
\begin{split}
\esp{I^\tau_t   - I^\tau_s   \Big{|} \widetilde \G_s}   & =  \int_{s}^{t}  \esp{\esp{{\mu(u, S^\tau_u, X_u^\tau) \over \sigma(u,S^\tau_u)} \I_{\{\tau > u\}}   - {\phantom{\Big{|}}}^{o,\widetilde \bG} \left( { \mu_u \over \sigma_u}  \I_{\{\tau > u\}}\right)  \Big{|}  \widetilde \G_u}\Big{|} \widetilde \G_s} \ud u  \\
& \qquad +  \esp{\esp{W^\tau_t  - W^\tau_s \Big{|}  \G_s}\Big{|} \widetilde \G_s}. \end{split}
\end{equation}
Since $\bE [W^\tau_t  - W^\tau_s |    \G_s ] = 0$,  finally we get
$$
\esp{I^\tau_t   - I^\tau_s   \Big{|} \widetilde \G_s}  =
\int_{s}^{t}  \esp{{\phantom{\Big{|}}}^{o,\widetilde \bG} \left( \frac{ \mu_u}{ \sigma_u}  \I_{\{\tau > u\}}   \right)  -   {\phantom{\Big{|}}}^{o,\widetilde \bG} \left( \frac{ \mu_u}{ \sigma_u}  \I_{\{\tau > u\}}   \right)}\ud u = 0.$$
To conclude, we apply the L\'evy Theorem taking into account that $\langle I^\tau \rangle =  \langle W^\tau \rangle$.
\end{proof}

\begin{proof}[Proof of Lemma \ref{markov}]
Recall that the process $\widehat W$ given in
\eqref{eq:What} and $B$ are
independent $(\bF, \widehat \P)$-Brownian motions.
Since the change of probability measure from $\P$ to $\widehat \P$ is Markovian, the pair $(S,X)$ is still an $(\bF,\widehat \P)$-Markov process, see~\citet[Proposition 3.4]{ceci2001nonlinear}.
Then, the Markov generator $\widehat \L^{S,X}$ of the pair $(S,X)$ can be easily computed by applying It\^{o}'s formula to the function $f(t,s,x)$.
\end{proof}

\begin{proof}[Proof of Lemma \ref{markovbis}]
Note that, for any bounded and measurable  function $f:\{0,1\} \to \R$ we get, for any $t \in [0,T]$
\begin{align}
f(H_t) =& \int_0^t ( f(H_{s^-} +1) - f(H_{s^-} ) )\ud H_s\\
=&\int_0^t ( f(H_{s^-} +1) - f(H_{s^-} ) ) \ud M_s  + \int_0^t ( f(H_{s^-} +1) - f(H_{s^-} ) (1-H_s) \gamma(s,X_s) \ud s,\end{align}
 and  since

 $$S^\tau _t  = \int_0^t S_r \ud H_r + S_t (1-H_t),  \quad  X^\tau _t  = \int_0^t X_r \ud H_r + X_t (1-H_t), \forall t \in [0,T]$$

 we obtain $ \ds \ud S^\tau _t = (1 - H_{t^-})  \ud S_t$ and  $\ds \ud X^\tau _t = (1 - H_{t^-})  \ud X_t$, $\forall t \in [0,T]$ .

Then, the Markov generator $\widehat \L^{S,X,H}$ is obtained by It\^{o}'s formula applied to any function $f(t,s,x,z) \in \widehat C_b^{1,2,2}$.
\end{proof}

\begin{proof}[Proof of Proposition \ref{filteq}]
First, observe that $\widehat W$ is an $(\bF^S, \widehat \P)$ Brownian motion since the following equality holds
$$
\widehat W_t = I^{\bF^S} _t + \int_0^t { {}^{p, \bF^S}\!\mu_u\over \sigma(u,S_u) } \ud u,\quad t \in [0,T],
$$
where  $\{I^{\bF^S}_t := W_t  + \int_0^{t }  {\mu(u, S_u, X_u)  -  {}^{p, \bF^S}\!\mu_u \over \sigma(u,S_u) } \ud u,\ t \in [0,T]\}$ is the so-called innovation process which is known to be an $(\bF^S, \widehat \P)$ Brownian motion (see, for instance \citet{liptser2013statistics}).

Recalling the semimartingale decomposition  of $f(t,S_t,X_t,Y_t)$, given in \eqref{semi},  we can proceed as in the proof of \citet[Proposition A.2]{ceci2015local} and prove that the filter $\pi$ solves equation  \eqref{eq:ks}.

Strong uniqueness for the solution to the filtering  equation follows by uniqueness of the {\em filtered martingale problem}  for the operator $\widehat \L^{S,X,Y}$ (see, e.g. \citet{kurtz1988unique}, \citet{cecicolaneri2012}, \citet{ceci2014zakai}).
Precisely, by applying \citet[Theorem 3.3]{kurtz1988unique} we get that the filtered martingale problem for the operator $\widehat \L^{S,X,Y}$ has a unique solution, and this implies uniqueness of equation \eqref{eq:ks}.
\end{proof}

\bibliographystyle{plainnat}
\bibliography{bibliography}

\end{document}